\documentclass[final]{siamltex}

\usepackage{amsmath}
\usepackage{amssymb}
\usepackage{graphicx}
\usepackage{epstopdf}
\usepackage{algorithm}
\usepackage{url}
\usepackage{cite}
\usepackage{color}
\usepackage{algorithmic,color}
\usepackage[top=1in, bottom=1in, left=1.25in, right=1.25in]{geometry}

\DeclareMathOperator*{\argmax}{arg\,max}
\DeclareMathOperator*{\argmin}{arg\,min}

\title{Total Variation and Tight Frame Image Segmentation with Intensity Inhomogeneity}

\begin{document}
\author{
      Raymond Chan\thanks{
       Department of Mathematics, City University of Hong Kong,
       Tat Chee Avenue, Hong Kong. Research supported by HKRGC Grants No. CityU12500915, CityU14306316, HKRGC CRF Grant C1007-15G, and HKRGC AoE Grant AoE/M-05/12. Email: rchan.sci@cityu.edu.hk.
         }
       \and
       Hongfei Yang\thanks{
       Department of Mathematics, The Chinese University of Hong Kong,
       Shatin, Hong Kong. Email: hongfeiyang@cuhk.edu.hk.
       }
       \and
       Tieyong Zeng\thanks{
        Department of Mathematics, The Chinese University of Hong Kong,
        Shatin, Hong Kong. Research supported by National Science Foundation of China No. 11671002, CUHK start-up and CUHK DAG 4053296, 4053342.
        Email:zeng@math.cuhk.edu.hk.}
        }
\maketitle


\providecommand{\norm}[1]{\lVert#1\rVert}
\providecommand{\abs}[1]{\lvert#1\rvert}

\newcommand{\red}[1]{{\color[rgb]{1,0,0}{#1}}}
\newcommand{\blue}[1]{{\color[rgb]{0,0,1}{#1}}}

\makeatletter
\def\moverlay{\mathpalette\mov@rlay}
\def\mov@rlay#1#2{\leavevmode\vtop{%
   \baselineskip\z@skip \lineskiplimit-\maxdimen
   \ialign{\hfil$\m@th#1##$\hfil\cr#2\crcr}}}
\newcommand{\charfusion}[3][\mathord]{
    #1{\ifx#1\mathop\vphantom{#2}\fi
        \mathpalette\mov@rlay{#2\cr#3}
      }
    \ifx#1\mathop\expandafter\displaylimits\fi}
\makeatother

\newcommand{\cupdot}{\charfusion[\mathbin]{\cup}{\cdot}}
\newcommand{\bigcupdot}{\charfusion[\mathop]{\bigcup}{\cdot}}

\begin{abstract}
Image segmentation is an important task in the domain of computer vision and medical imaging. In natural and medical images, intensity inhomogeneity, i.e. the varying image intensity, occurs often and it poses considerable challenges for image segmentation. In this paper, we propose an efficient variational method for segmenting images with intensity inhomogeneity. The method is inspired by previous works on two-stage segmentation and variational Retinex. Our method consists of two stages. In the first stage, we decouple the image into reflection and illumination parts by solving a convex energy minimization model with either total variation or tight-frame regularisation. In the second stage, we segment the original image by thresholding on the reflection part, and the inhomogeneous intensity is estimated by the smoothly varying illumination part. We adopt a primal dual algorithm to solve the convex model in the first stage, and the convergence is guaranteed. Numerical experiments clearly show that our method is robust and efficient to segment both natural and medical images.
\end{abstract}


\pagestyle{myheadings}
\thispagestyle{plain}
\markboth{}{IMAGE SEGMENTATION WITH INTENSITY INHOMOGENEITY}

\section{Introduction}
The main goal of image segmentation is to partition the underlying image into different nearly homogeneous segments. However, intensity inhomogeneity, which occurs often in natural and medical images, \cite{BAR1998,PDL1998,RJC1998,Li_MRI,hybridclustering,VU2006} will create significant challenges for image segmentation. Here, by \emph{intensity inhomogeneity} (or {\it biased field}), we refer to the spurious smoothly varying image intensities \cite{VU2006,AMN2002,KSG,ZYC2013}.  Many efficient segmentation methods such as \cite{CV,YBTBmul,YBTBtwo,DCS,LNZS,CXiaoCRaymZTiey1} assume that different phases of an image can be well approximated by constant functions. Therefore, if an object in the image has varying intensities, or the intensities of different objects have overlaps, these methods may fail to give good segmentation results. Fig.~\ref{fishdec} is a good example to illustrate the concept of intensity inhomogeneity---the background of this image is not uniform. Fig.~\ref{fishdec}(b) is an estimation of the illumination part of the image given by our method. From it, we see clearly that the bottom part of the image is  darker, while the top part of the image is lighter. Naturally, the illumination part of an image can provide environment information for image analysis, but the spatially varying light in the background, such as that in Fig.~\ref{fishdec}(a), makes the segmentation challenging, see our later experiment in Fig.~\ref{Fish}.

\begin{figure}[htbp]
\begin{center}
\begin{minipage}[t]{4.5 cm}
\includegraphics[height=3.35 cm]{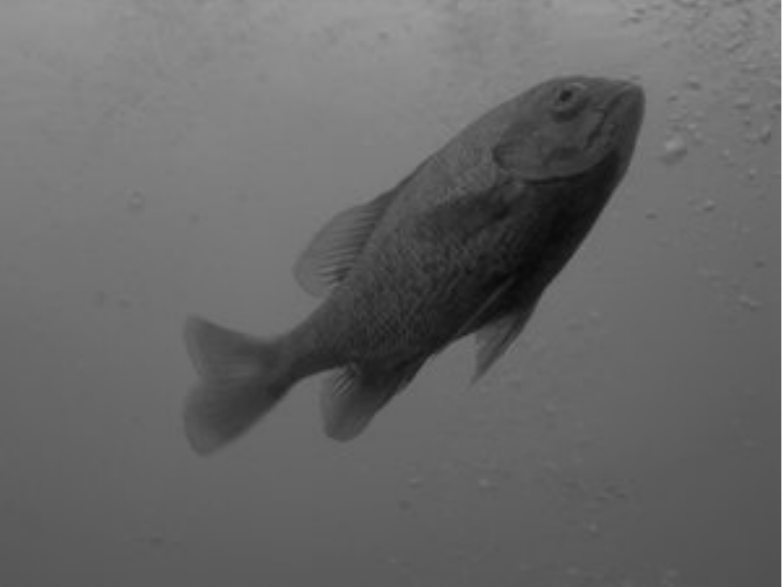}
\centering{(a) Original image}
\end{minipage}
\begin{minipage}[t]{4.5 cm}
\includegraphics[height=3.35 cm]{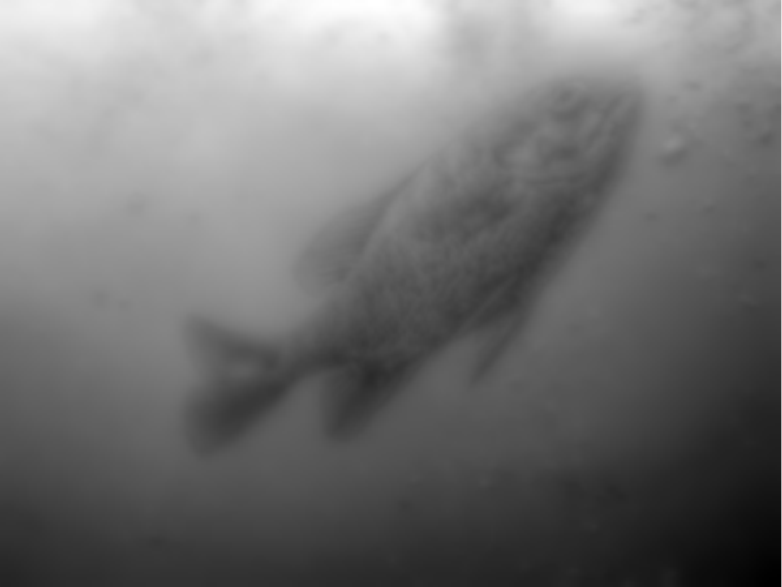}
\centering{(b) The illumination}
\end{minipage}
\end{center}
\caption{Example of intensity inhomogeneity. (a) Original image ``Fish'', (b) The estimated environment illumination from (a) by our method.}
\label{fishdec}
\end{figure}

In the literature, several methods have been proposed for image segmentation with intensity inhomogeneity. In \cite{BAR1998,JSMA1996}, the estimation of intensity inhomogeneity is through certain low pass filter. In \cite{ultrasound_seg, LiX, RJC1998,KSG}, the authors proposed to compute intensity inhomogeneity by modified finite Gaussian mixture models. In \cite{AMN2002,PDL1998}, the authors modified the fuzzy C-means algorithm to deal with intensity inhomogeneity. For the model in \cite{AMN2002}, the labeling of a pixel is decided by fuzzy C-means algorithm and the labeling of its immediate neighborhood. For the model in \cite{PDL1998}, a varying field is multiplied to the centroids of the clusters, and a regularisation of the varying field is added. In \cite{ZYC2013}, a multi-phase segmentation method is proposed based on the MAP principle and local information of the joint density. In their model, intensity inhomogeneity at each pixel is estimated from the neighborhood intensities, and a weight function based on the distance to the pixel is applied in the estimation procedure. In \cite{hybridclustering}, the authors combined hard, fuzzy and probabilistic criteria into a hybrid clustering algorithm, and smoothing filters are also incorporated to reduce noise and to ensure the smoothness of the bias field. In \cite{Li_MRI}, the authors used the localized K-means clustering method and the level set approach to segment images with intensity inhomogeneity. In their model, intensity inhomogeneity at small neighborhood is assumed to be constant and the local different phases can be clustered by K-means method. In \cite{Kaihua2016level}, the authors modeled inhomogeneous objects as Gaussian distributions of different means and variances. They used a sliding window to transform the original image domain into another domain, where the different Gaussian distributions can be better separated. A level set method with regularisation is then implemented to segment the transformed image. For a comprehensive review on segmentation with intensity inhomogeneity, please refer to \cite{VU2006}. We note that most of the methods mentioned above are designed, or extensively tested, on medical images. We stress that our proposed method performs well on both medical images and natural images.

In order to get better segmentation results for images with intensity inhomogeneity, one natural idea is to first remove the non-uniform field from the underlying image, and then segment on the remained image. In this paper, we are particularly interested in the Retinex theory \cite{TVRetinex,morel-TIP2010,MaOsher2011} which turns out to be an extremely important tool to remove intensity inhomogeneity. The term \emph{Retinex}, first coined in \cite{LEH1971}, is combined from the words \emph{retina} and \emph{cortex}. The Retinex theory explains how human eyes perceive constant colors under various illuminations.  In \cite{TVRetinex}, Ng and Wang proposed a novel variational Retinex approach to improve qualities of images with intensity inhomogeneity. In their model, the observed image $S$ is determined by illumination $L$ and reflection $R$ of the underlying objects in the following way
\begin{equation}\label{slr}
S= L\cdot R,
\end{equation}
where physically $0< R\leq 1$.
Further assuming that the illumination $L$ is spatially smooth and the reflection $R$ is piecewise constant, in order to extract $L$ and $R$ from $S$, they proposed to consider the following minimization problem
\begin{equation} \label{TV_Retinex}
    \inf_{r\geq 0,l} \left\{ \int_\Omega \abs{Dr} + \frac{\beta}{2}\int_\Omega \abs{\nabla l}^2dx + \frac{\gamma}{2}\int_\Omega (l-r-s)^2dx + \frac{\mu}{2} \int_\Omega l^2dx \right\},
\end{equation} where $s=\log(S)$, $r=-\log(R)$ and $l=\log(L)$. Note that here $Dr$ should be understood in the distributional sense, and the parameters $\beta$, $\gamma$ and $\mu$ are positive.
As explained in \cite{TVRetinex}, the last term here is to ensure the well-posedness of the model, and typically $\mu$ can be taken to be very small. After \eqref{TV_Retinex} is solved, one can perform a Gamma correction on $L$ to get $L'$ and set $S' = L'\cdot R$. Numerical experiments in \cite{TVRetinex} show that the modified image $S'$ has better quality with less intensity inhomogeneity.

For image segmentation, one of the most prominent approaches was given in \cite{MS-pre,MS} by Mumford and Shah. In these seminal works, they proposed to segment the image $S$ by calculating an optimal approximation $u$ of $S$ and a decomposition of the image domain
\begin{equation*}
	\Omega = \Omega_1 \cup \Omega_2 \cup \cdots \cup \Omega_n \cup \Gamma,
\end{equation*}
where $\Omega_i$'s are connected open disjoint subsets of $\Omega$, and $\Gamma$ is the collection of the boundaries of $\Omega_i$'s in $\Omega$. As $u$ is required to be continuous and to approximate $S$ in $\Omega_i$, the objective functional to be minimised in \cite{MS-pre,MS} is then given by
\begin{equation}
	\label{MSmodel}
	E(u,\Gamma):=\text{Length}(\Gamma)+\frac{\beta}{2}\int_{ \Omega \backslash \Gamma} \abs{\nabla u}^2 dx+\frac{\lambda}{2}\int_\Omega (S-u)^2 dx  ,
\end{equation} where $\lambda$ and $\beta$ are positive parameters and the length of $\Gamma$ can be written as $\mathcal{H}^{1}(\Gamma)$, the $1$-dimensional Hausdorff measure in $\mathbb{R}^2$, see \cite{TonyMumShah}.

Due to the non-convexity of the term $\text{Length}(\Gamma)$, the minimization of the energy \eqref{MSmodel} is extremely challenging, see \cite{C,CA,AT,AT2,TonyMumShah,David2005,CXiaoCRaymZTiey1,PCBC,PCCB,DMN2008} and reference therein for some previous effective efforts. In \cite{CXiaoCRaymZTiey1}, the authors proposed a novel two-stage segmentation method, which is closely related to the original Mumford-Shah model. Indeed, in the first stage of their approach, they proposed to solve the following convex minimization problem
\begin{equation}
\label{CaiModel}
	\inf_{u\in W^{1,2}(\Omega)}{E(u)}:= \int_\Omega \abs{\nabla u}dx+\frac{\beta}{2}\int_\Omega \abs{\nabla u}^2 dx +\frac{\lambda}{2}\int_\Omega (S-Au)^2dx,
\end{equation} where $A$ is a blurring kernel if the given image is blurred by $A$.
Here, similarly to \cite{CEN,LNZS}, the first term is to control the length of edges in the solution image,
the second term is for smoothing the image to erase tiny structures\cite{CXiaoCRaymZTiey1,Chanyangzeng}, and the third term is the classical data-fidelity term. Note that the model \eqref{CaiModel} closely connects three major tasks in image processing: denoising, deblurring and segmentation
and it has been utilized in \cite{Muller} for image restoration.

After solving \eqref{CaiModel}, the second stage in \cite{CXiaoCRaymZTiey1} is to segment the original image $S$ by thresholding on the solution image $u$. This two-stage approach has several advantages: first, the minimization problem \eqref{CaiModel} is convex and there exists fast numerical scheme to solve the minimization problem; second, this two-stage approach can solve multi-phase segmentation efficiently; third, the thresholding in the second stage is independent of solving the first stage, and users can try different thresholds and/or number of phases without recalculating \eqref{CaiModel}. Despite the superior numerical performance of this two-stage approach \cite{CXiaoCRaymZTiey1,Chanyangzeng}, it is still an open question to understand the mathematical connection and difference between (\ref{CaiModel}) and the Mumford-Shah model \eqref{MSmodel}. Indeed, Cai and Steidl showed that their two-class segmentation model by the so called Iterated ROF Thresholding procedure is equivalent to the Chan-Vese model with some adapted regularisation parameter, see \cite{Cai-steidl2013} for more details. It is well-known that the Chan-Vese model is a simplified version of \eqref{MSmodel} (taking $\mu=+\infty$) and the ROF model is a particular case of \eqref{CaiModel}. Moreover, in \cite{Chanyangzeng}, the authors illustrated that for a class of simple images, theoretically (\ref{CaiModel}) yields the same solutions as the Mumford-Shah model.

The current paper explores another important aspect, i.e., the above-mentioned intensity inhomogeneity issue which affects significantly image segmentation results. Indeed, inspired by the variational Retinex approach \cite{TVRetinex} and the two-stage segmentation method \cite{CXiaoCRaymZTiey1}, in this paper we propose a two-stage method to segment images with intensity inhomogeneity. In the first stage, we decouple the observed image $S$ into illumination $L$ and reflection $R$. This is done by solving a convex minimization problem with an extra smoothing term on the reflection $R$ (precisely, on $r=-\log R$ in \eqref{mainmodel} below), which will be utilized for the thresholding step in the second step. To achieve a balance between computational speed and fine details/boundaries in the segmentation, we propose to implement our model with either TV or tight frame regularisation in our first stage, which will be useful to understand the connection and difference between TV and the tight frame framework, as explored in \cite{CDOS2012}.

The contribution of this paper are the followings. First, we propose a new variational model to segment images with intensity inhomogeneity. Our model combines the variational Retinex model \eqref{TV_Retinex} with an extra smoothing term and the two-stage method to get a segmentation. Note that the target of our paper is different from \cite{TVRetinex} since theirs is for image enhancement. Secondly, we demonstrate how to employ the Chambolle-Pock algorithm \cite{CP} to solve the minimization problems we proposed. In our algorithm, there is only one loop and every updates of the variables are exact. Convergence of our method is ensured and the convergence rate is also known and could be improved. However, the numerical scheme in \cite{TVRetinex} contains an inner loop, and some updates of the variables are thus not exact and the theoretical convergence of their method is unknown, or at least, needs some extra work. In this regard, our numerical scheme serves as a good illustrative template to avoid such problem.

The rest of the paper is organized as follows. In Section \ref{sec:background}, we give a brief introduction to the theory of tight-frame regularisation. In Section \ref{sec:proposed_model}, we formulate and discuss our proposed model. In Section \ref{sec:primal_dual}, we propose to solve the first stage of our model in its primal-dual form by the Chambolle-Pock algorithm. In Section \ref{sec:numerical_experiments}, we numerically compare our approach with some other methods \cite{CV,Li_MRI,YBTBmul,YBTBtwo,Kaihua2016level}. In the last section, we conclude the discussion with possible future improvements.

\section{Tight frame regularisation} \label{sec:background}

In this section, we briefly introduce the tight-frame theory used in the segmentation model with frame based regularisation \cite{DCS}. The following introduction summarizes materials from \cite{DCS,SZ2010}. Readers interested in the theory of tight-frames and framelets can consult \cite{DI1992, RA1997, SZ2010,DB2010, DCS}.

A tight frame of $L_2(\mathbb{R})$ is a countable set $X\subset L_2(\mathbb{R})$ satisfying
\begin{equation*}
f = \sum_{h\in X}\langle f,h \rangle h, \quad \forall f\in L_2(\mathbb{R}),
\end{equation*}
where $\langle \cdot,\cdot \rangle$ is the inner product of $L_2(\mathbb{R})$. For given $\Psi = \{ \psi_1,\cdots,\psi_v\}\subset L_2(\mathbb{R})$, we call the collection of the dilations and the shifts of $\Psi$
\begin{equation} \label{tight_frame_system}
X(\Psi)=\{ \psi_{t,j,k}:1\leq t \leq v, j,k\in\mathbb{Z}\} \quad \text{with} \quad \psi_{t,j,k}=2^{j/2}\psi_t(2^j\cdot-k)
\end{equation} an affine system. We call $X(\Psi)$ a tight wavelet frame, and call $\psi_t,t=1,\cdots,v$ the (tight) framelets, when $X(\Psi)$ forms a tight frame of $L_2(\mathbb{R})$.

To construct a set of framelets from multiresolution analysis, one usually starts from a compactly supported refinable function $\phi\in L_2(\mathbb{R})$ (a scaling function) with a finitely supported sequence $h_0\in \ell_2(\mathbb{Z})$ (a refinement mask) satisfying
\begin{equation*}
\phi(x) = 2\sum_{k\in\mathbb{Z}} h_0(k)\phi(2x-k),
\end{equation*} or in the Fourier domain
\begin{equation*}
\widehat{\phi}(2\cdot)=\widehat{h}_0\widehat{\phi}.
\end{equation*} Here $\widehat{\phi}$ is the Fourier transform of $\phi$, and $\widehat{h}_0$ is the Fourier series of $h_0$ defined by
\begin{equation*}
\widehat{h}_0(\omega) = \sum_{k\in\mathbb{Z}} h_0(k)e^{-ik\omega}, \quad \omega \in \mathbb{R}.
\end{equation*}
It can be seen that $\widehat{h}_0(0)=1$, which means that a refinement mask of a refinable function must be a low pass filter. For a given compactly supported refinable function, the construction of a tight framelet system is to find a finite set $\Psi$ whose elements can be represented by
\begin{equation*}
\psi_t(x) = 2\sum_{k\in \mathbb{Z}} h_t(k)\phi(2x-k)
\end{equation*}
with finite supported sequence $h_t\in\ell_2(\mathbb{Z})$ with $2\pi$-periodic Fourier series, or in the Fourier domain
\begin{equation*}
\widehat{\psi}_t(2\cdot)=\widehat{h}_t\widehat{\phi}.
\end{equation*}
According to the unitary extension principle (UEP) \cite{RA1997}, the tight framelet system $X(\Psi)$ in \eqref{tight_frame_system} generated by $\Psi$ forms a tight frame in $L_2(\mathbb{R})$  provided that the masks $\widehat{h}_t$ for $t=0,1,\cdots,v$ satisfy
\begin{equation*}
\sum_{t=1}^v \widehat{h}_t(\xi)\overline{\widehat{h}_t(\xi+\gamma\pi)} = \delta_{\gamma,0}, \quad \gamma = 0,1,
\end{equation*} for almost all $\xi\in\mathbb{R}$. While $h_0$ corresponds to a low pass filter, $\{h_t:t=1,2,\cdots,v\}$ must correspond to high pass filters by the UEP. In our implementation, we adopt the piecewise linear B-spline framelet. The corresponding refinable function is $\phi(x) = \max\{1-\abs{x},0\}$, and the refinement mask is $\widehat{h}_0(\xi)=\cos^2(\frac\xi 2)$, and two framelets $\psi_1$ and $\psi_2$ are determined by $\widehat{h}_1=-\frac{\sqrt{2}i}{2}\sin(\xi)$ and $\widehat{h}_2=\sin^2(\frac{\xi}2)$. The corresponding filters are
\begin{equation} \label{TFfilters}
h_0 = \frac{1}{4}[1,2,1], \quad h_1 = \frac{\sqrt{2}}{4}[1,0,-1], \quad h_2 = \frac{1}{4}[-1,2,-1].
\end{equation} Fig.~\ref{framelets} is a plot of the refinable function $\phi$ and the framelet functions $\psi_1$ and $\psi_2$.

\begin{figure*}[htbp]
\begin{center}
\begin{minipage}[t]{4 cm}
\includegraphics[height=2.8 cm]{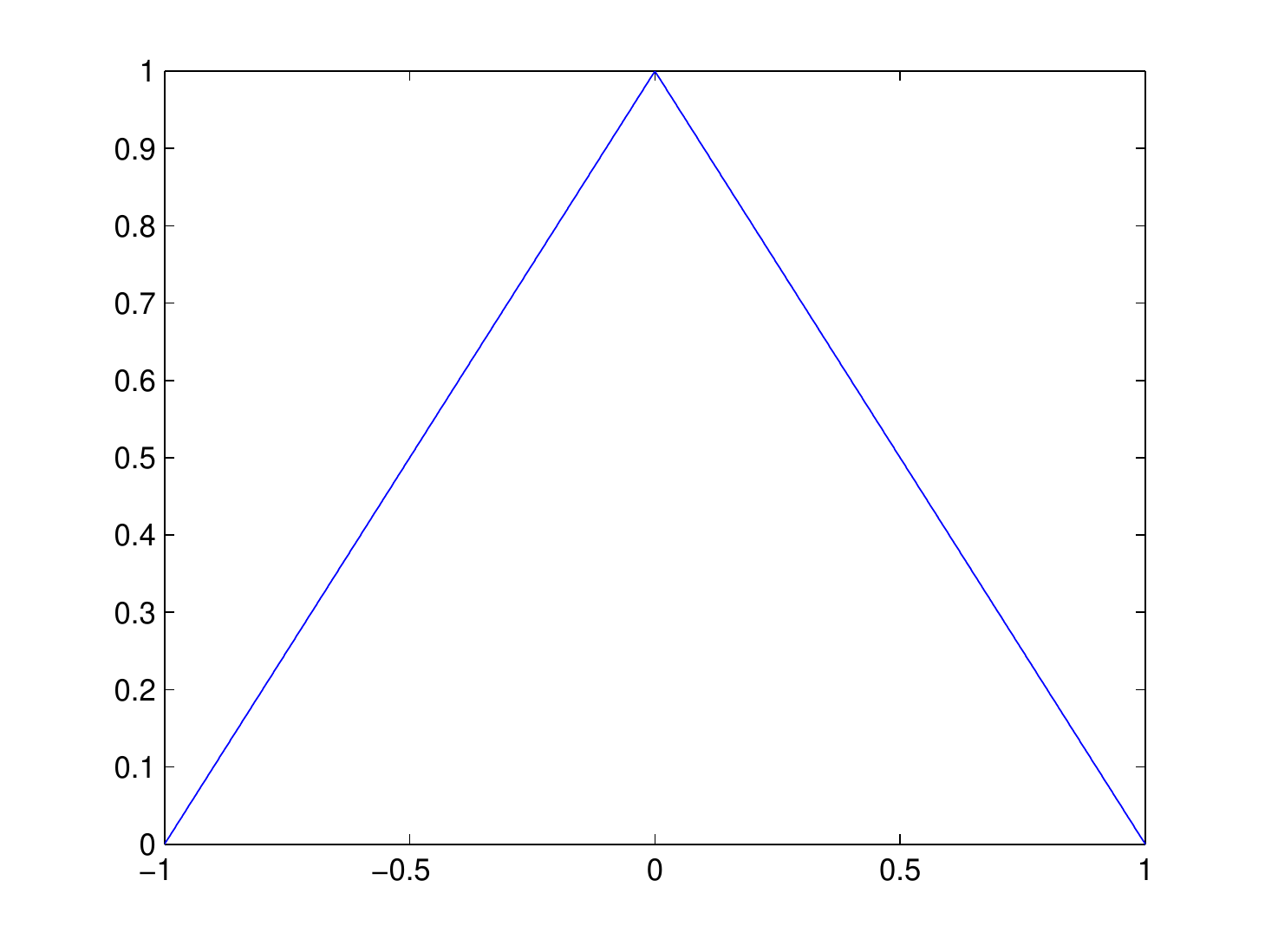}\\
\centering{(a) Refinable function $\phi$}
\end{minipage}
\begin{minipage}[t]{4 cm}
\includegraphics[height=2.8 cm]{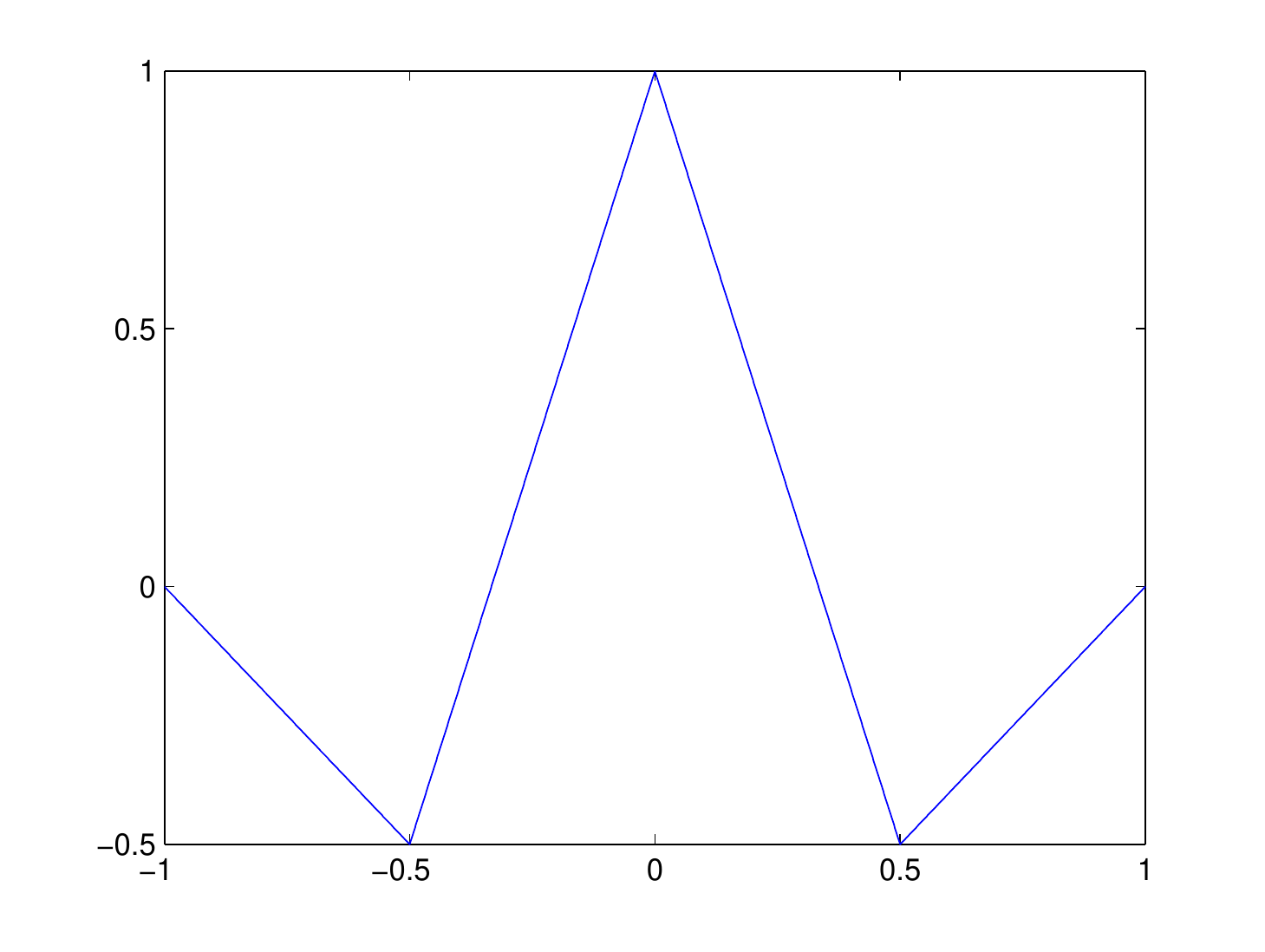}\\
\centering{(b) $\psi_1$}
\end{minipage}
\begin{minipage}[t]{4 cm}
\includegraphics[height=2.8 cm]{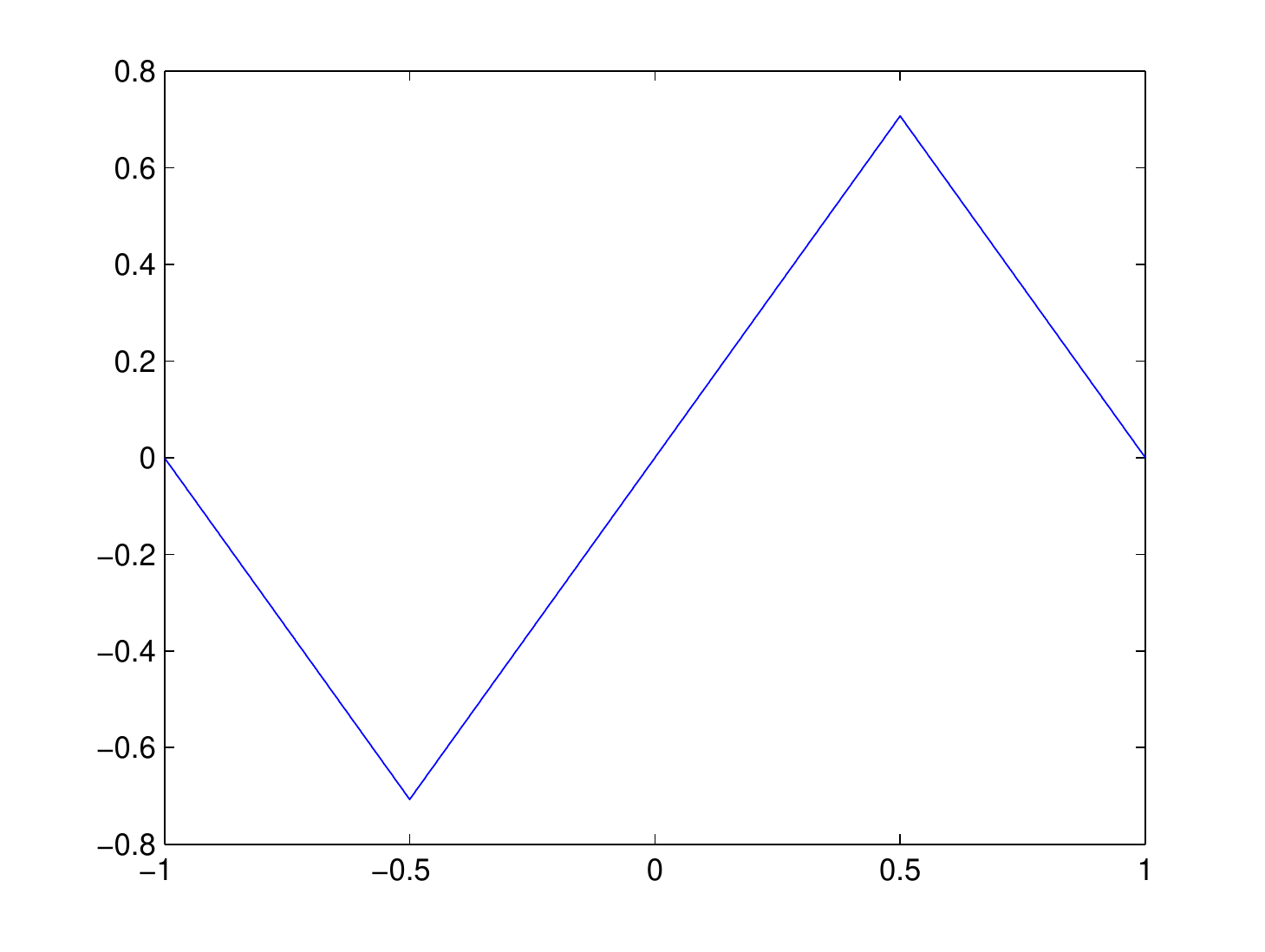}\\
\centering{(c) $\psi_2$}
\end{minipage}
\end{center}
\caption{\label{framelets}}
\end{figure*}

The $d$-dimensional framelet system for $L_2(\mathbb{R}^d)$ can be constructed by tensor products of one-dimensional framelets. Indeed, if we have one scaling function $\phi$ and $v$ tight framelets $\psi_1,\cdots, \psi_v$ in $1$D, then after tensor product, we obtain a tight frame system generated by one scaling function and $(v+1)^d-1$ tight framelets.

In the discrete setting, we regard a discrete image $f$ as the coefficients $\{f_i=\langle f_c,\phi(\cdot -i)\rangle \}$ up to a dilation, where $f_c$ is the continuous function,  $\phi$ is the refinable function associated with the framelet system, and $\langle \cdot,\cdot \rangle$ is the inner product in $L_2(\mathbb{R}^d)$. The $K$-level discrete framelet decomposition of $f$ is then the coefficients $\{f_i=\langle f,2^{-K/2}\phi(2^{-K}\cdot -i)\rangle \}$ at a prescribed coarsest level $K$, and the framelet coefficients are
\begin{equation*}
\{ \langle f, 2^{-k/2}\psi_t(2^{-k}\cdot -j)\rangle, \quad 1\leq t \leq (v+1)^d-1\}
\end{equation*} for $1\leq k\leq K$.

For a discrete $d$-dimensional image, we can concatenate it column-wise to a vector in $\mathbb{R}^n$, where $n$ is the total number of pixels in the image. Then the framelet decomposition and reconstruction can be represented by matrix multiplications $Wf$ and $W^T \eta$ ($\eta$ is the framelet coefficients) respectively, where $W\in \mathbb{R}^{m\times n}$ with $m=n(v+1)^d$, and $W$ satisfies the ``\emph{perfect reconstruction property}'' $W^TW = I$ by UEP. The matrix $W$ has the form
\begin{equation}
W = \left( \begin{array}{c}
H_0 \\
H_{1,1} \\
H_{1,2} \\
\vdots  \\
H_{K,(v+1)^d-1}
\end{array} \right),
\end{equation} where $H_0$ is the submatrix of $W$ corresponding to the decomposition with respect to the refinable function, and $H_{k,t}$ with $1\leq k \leq K$ and $1\leq t \leq (v+1)^d-1$ are the submatrix of $W$ corresponding to the decomposition at the $k$-th level with respect to the $t$-th framelet. For example, for one level decomposition with the piecewise linear B-spline framelet in 2-dimension, the matrix $W$ is of size $9n\times n$.

As noted in \cite{DCS}, frame based regularisation mainly has two advantages over the TV regularisation. First, piecewise smooth functions have sparser representations under tight-frame systems, and penalizing $\ell_1$-norm of $Wu$ generally should generate better results than penalizing $\norm{\nabla u}_1$ as confirmed by researches in image restoration problems \cite{COS2009, CD2004,CCPW2007, DTV2007,ESQD2005,FSM2009}. Secondly, $\norm{Wu}_1$ contains more geometric information of the image $u$ than $\norm{\nabla u}_1$ since it contains more filtering directions.

\section{The proposed model} \label{sec:proposed_model}
Our proposed model has two stages. In the first stage, we aim to decouple a target image into illumination and reflection parts by changing the model in the variational Retinex  approach \cite{TVRetinex}. For most images the intensity of a point is determined by two parts, the intensity of the illumination received at this point, and the ratio of the illumination reflected from this point \cite{HBK1977}. Assuming that for one object the ability to reflect illumination is homogeneous, it is natural to model the observed intensity inhomogeneity to be induced by varying illumination \cite{LZ1994}. The varying illumination may have different sources: the varying distances from a point lighting source, the uneven thickness of clouds that filter sun lights, or the biased magnetic field generated in an MRI machine. Regardless of the sources, we assume that illumination should vary smoothly \cite{TVRetinex}.

Let us consider in the discrete setting. Denote $\Omega$ to be the image domain. As in \cite{TVRetinex}, assume that an image $S\in\mathbb{R}^{MN\times1}$ (size $M\times N$, concatenated columnise) is decomposed as
\begin{equation} \label{cont_image_model}
S = L\cdot R,
\end{equation} where $L$ represents illumination and $R$ represents reflection, and the multiplication is entrywise. The above formula can represent different modalities of image acquisition: for photos taken by optical cameras, $L$ represents the intensity of light, and $R$ represents the reflectivity of underlying objects; for magnetic resonance imaging (MRI), $L$ represents the intensity of magnetic field, and $R$ represents the susceptibility of organs. Similar to \cite{TVRetinex}, we assume that intensity inhomogeneity comes entirely from the illumination $L$ and $0< R\leq 1$, where $R \to 0$ means absolute absorption of the incident illumination and $R = 1$ means absolute reflection. We also assume that the illumination has the range $0< L<\infty$. Taking logarithm on both sides of (\ref{cont_image_model}) and denoting $s = \log(S)$, $l = \log(L)$ and $r = -\log(R)$, we get
\begin{equation} \label{con_image_model_log}
l = s+r.
\end{equation}

To obtain a satisfactory segmentation for the image $S$, in the first stage we estimate the inhomogeneous illumination by decoupling $R=\exp{(-r)}$ and $L=\exp{(l)}$. In the second stage, we get a segmentation of $S$ by thresholding the reflection $R$. In order to separate $r$ from $l$, we combine the variational Retinex model \eqref{TV_Retinex} with the two-stage segmentation model \eqref{CaiModel} (discrete version) to form the following convex optimization problem
\begin{equation} \label{mainmodel}
\min_{r\geq 0,l}\left\{ \norm{GWr}_1 + \frac{\alpha}{2}\norm{\nabla r}_2^2+\frac{\beta}{2}\norm{\nabla l}_2^2+ \frac{\gamma}{2}\norm{l-s-r}_2^2+\frac{\mu}{2}\norm{l}_2^2\right\},
\end{equation}
where $\alpha, \beta,\gamma$ and $\mu$ are positive parameters and similarly to \eqref{TV_Retinex}, $\mu$ will be a fixed small number. The norm $\norm{\cdot}_2$ denotes the $\ell_2$ norm, while the norm $\norm{\cdot}_1$ is the Isotropic TV defined in \cite{TVRetinex}. The term $\norm{GWr}_1$ is to control the length of edges in $r$, see \cite{DCS,TZS2013}. For our implementation with tight frame regularisation, the linear operator $W$ represents the matrix of 1-level framelet decomposition with respect to the piecewise linear B-spline framelet constructed in \cite{RA1997}. In this case, the image is in $\mathbb{R}^{MN\times1}$, so $W$ is a $9MN\times MN$ matrix of the form
\begin{equation*}
    W =
    \begin{pmatrix}
    H_0 \\
    H_1 \\
    \vdots \\
    H_8
    \end{pmatrix}
\end{equation*}
generated by tensor products of the filters \eqref{TFfilters}. As in \cite{DCS} the matrix $G$ is a square positive diagonal weight matrix defined by
\begin{equation*}
G = \diag{\{\mathbf{0},v_1,v_2 \cdots v_8\}},
\end{equation*}
with $\mathbf{0}, v_i \in \mathbb{R}^{1\times MN}$ and
\begin{equation} \label{weight_matrix}
v_i(j) = v(j) = \frac{1}{1+\epsilon\sum_{k=1}^8\abs{(H_k\tilde{s})(j)}^2}.
\end{equation}
Here $\epsilon=50/(MN)$, and $\tilde{s}$ is a smoothed version of $s = \log S$. In our numerical implementation we smooth $s$ by filtering it with a Gaussian kernel with variance $1$. Notice that $G$ can be regarded as the edge indicator function under the framelet transform $W$, see \cite{DCS}. Similar to \cite{ChamboJMIV,CP}, the discrete gradient operator $\nabla$ has the form
\begin{equation}
    \nabla =
    \begin{pmatrix}
    \nabla_x \\
    \nabla_y
    \end{pmatrix},
\end{equation} where $\nabla_x, \nabla_y \in \mathbb{R}^{MN\times MN}$ represents the horizontal and vertical finite difference of the discrete image.  For example, we define
\begin{equation}
(\nabla_x u)_{i,j} = \left\{ \begin{array}{cl} u_{i+1,j}-u_{i,j} & i = 1,2,\cdots,M-1 \\ 0 & i=N. \end{array}\right.
\end{equation}
As in \cite{TVRetinex}, the term $\norm{l}_2^2$ guarantees that \eqref{mainmodel} is strictly convex and there exists a unique minimizer.

In \eqref{mainmodel}, we use weighted tight frame regularisation to get a convex model.
When $G$ is the identity matrix and $W=\nabla$, it becomes the TV regularisation model:
\begin{equation} \label{mainmodelTV}
    \min_{r\geq 0,l}\left\{ \norm{\nabla r}_1 + \frac{\alpha}{2}\norm{\nabla r}_2^2+\frac{\beta}{2}\norm{\nabla l}_2^2+ \frac{\gamma}{2}\norm{l-s-r}_2^2+\frac{\mu}{2}\norm{l}_2^2\right\}
\end{equation}
which is also convex. From numerical experiments we observe that the tight-frame model \eqref{mainmodel} can produce segmentation with finer details while the TV model requires less time.

Let us turn to the second stage: thresholding to get the segmentation result. As physically, the reflection better describes the objects in the image (see Fig. \ref{Fish}(g) and (h) for instance). Therefore the segmentation of the original image should be basically based on the reflection $R = \exp{\{-r\}}$. After obtaining $R$ from model \eqref{mainmodel} or \eqref{mainmodelTV}, for simplicity, in the second-stage we propose to get a segmentation of the original image $S$ by a simple thresholding on $R$.  For $K$ phase segmentation ($K\geq 2$), assume that we have re-scaled $R$ to have $0\leq R\leq1$, and have the $K+1$ thresholds $0=\rho_0<\rho_1<\rho_2< \cdots <\rho_{K-1}<\rho_{K}=1$. Then the $i$th phase is defined to be the pixels satisfying $\rho_{i-1}\leq R(x)<\rho_i$. Note that one does not need to recalculate the first stage when trying different number of phases or thresholds. Therefore changing the number of phases or thresholds does not cost extra computational time.

\section{The primal-dual algorithm for minimization}\label{sec:primal_dual}

Because of the convexity of the minimization problems \eqref{mainmodel} and \eqref{mainmodelTV}, many methods can be used to solve them. For example, the primal-dual algorithms \cite{ ChamboJMIV,CP,PDChan1, PDChan2}, which can be easily adapted to a number of non-smooth convex optimization problems and is easy to implement; the alternating direction method with multipliers (ADMM) \cite{BPCPE,FB} or Split-Bregman algorithm \cite{GO,CXiaoCRaymZTiey1}, which is convergent and well-suited to large-scale convex problems. In this section, we propose to solve our models \eqref{mainmodel} and \eqref{mainmodelTV} by the Chambolle-Pock algorithm \cite{CP,Chanyangzeng}, which belongs to one of the primal-dual algorithms. Since the algorithms to solve \eqref{mainmodel} and \eqref{mainmodelTV} are essentially the same, we only present the algorithm for \eqref{mainmodel} and we leave the details of solving \eqref{mainmodelTV} to interested readers.

First let us fix some notations. Recall that the images are in $\mathbb{R}^{MN\times1}$ (size $M\times N$, concatenated column-wise). For two vectors $u$, $v$ of the same size, we say $u\leq v$ if $u(i)\leq v(i)$ for all $i$. For a $9MN\times 1$ vector
\begin{equation*}
p = \begin{pmatrix} p_0 \\ \vdots \\ p_8 \end{pmatrix},
\end{equation*} denote $\abs{p}_2$ to be an $MN\times 1$ vector defined by
\begin{equation*}
\abs{p}_2 (i) = \sqrt{\sum_{k=0}^8 p_k(i)^2},\quad \forall i.
\end{equation*} We say $\widetilde{p}\in \mathbb{R}^{9MN\times 1}$ is a projection to the boxed constraint $\abs{p}_2\leq q$ with a positive vector $q\in \mathbb{R}^{MN\times 1}$ by defining
\begin{equation} \label{box_constraint}
\widetilde{p}_i(j) = \left\{ \begin{array}{rl} p_i(j) & \text{if } \abs{p}_2(j)\leq q(j), \\  q(j)p_i(j)/\abs{p}_2(j) & \text{if } \abs{p}_2(j)>q(j). \end{array} \right.
\end{equation}

Classically, the primal-dual formulation of \eqref{mainmodel} is given by
\begin{equation} \label{mainmodelPD}
\max_{\{\abs{p}_2 \leq \abs{\diag{G}}_2,q, u\}} \min_{r\geq 0, l} \left\{ \frac{\gamma}{2}\norm{l-s-r}_2^2 + \frac{\mu}{2}\norm{l}_2^2 + \langle Wr,p \rangle + \langle \nabla r,q \rangle + \langle \nabla l,u \rangle - \frac{1}{2\alpha} \norm{q}_2^2 - \frac{1}{2\beta}\norm{u}_2^2 \right\}.
\end{equation}

Denote
\begin{equation} \label{operator_def}
x = \begin{pmatrix} r \\ l \end{pmatrix},\quad y = \begin{pmatrix} p \\ q \\ u \end{pmatrix}, \quad K = \begin{pmatrix} W &  0 \\  \nabla & 0\\ 0 & \nabla \end{pmatrix},
\end{equation} and define
\begin{equation}
H(x) = \frac{\gamma}2 \norm{l-s-r}_2^2 + \frac{\mu}2 \norm{l}_2^2 + \delta(r), \quad F^*(y) = \frac{1}{2\alpha} \norm{q}_2^2 + \frac{1}{2\beta}\norm{u}^2 + \iota_G(p),
\end{equation} with
\begin{equation*}
\delta(r) = \left\{ \begin{array}{rl} 0 & \text{if } r_i\geq 0, \forall i \\ \infty & \text{otherwise}, \end{array} \right.
\end{equation*} and
\begin{equation*}
\iota_G(p) = \left\{ \begin{array}{rl} 0 & \text{if } \abs{p}_2 \leq \abs{\diag G}_2, \\ \infty & \text{otherwise }. \end{array} \right.
\end{equation*}
Then the primal-dual formulation \eqref{mainmodelPD} can be rewritten as
\begin{equation} \label{pd_general}
\min_x \max_y \left\{ H(x) + \langle Kx, y\rangle - F^*(y) \right\},
\end{equation}
which is exactly the same saddle point problem appearing in \cite{CP}.

Giving the initializations $(p_0,q_0,r_0,l_0,\bar{r}_0,\bar{l}_0)$, the Chambolle-Pock algorithm in \cite{CP} to solve for \eqref{pd_general} is thus given through the following iterations for $n\geq 0$,
\begin{equation} \label{updatep}
p_{n+1}=\argmax_{\abs{p}_2 \leq \abs{\diag{G}}_2} \left\{ \langle W\bar{r}_n,p\rangle - \frac{1}{2\tau}\norm{p-p_n}_2^2 \right\},
\end{equation}
\begin{equation} \label{updateq}
q_{n+1}=\argmax_q \left\{ \langle \nabla \bar{r}_n, q\rangle -\frac{1}{2\alpha}\norm{q}_2^2 -\frac{1}{2\tau}\norm{q-q_n}_2^2 \right\},
\end{equation}
\begin{equation} \label{updateu}
u_{n+1}=\argmax_u \left\{ \langle \nabla \bar{l}_n,u\rangle - \frac{1}{2\beta}\norm{u}_2^2-\frac{1}{2\tau}\norm{u-u_n}_2^2 \right\},
\end{equation}
\begin{multline} \label{updaterl}
(r_{n+1},l_{n+1})=\argmin_{r\geq 0,l} \Big\{ \frac\gamma 2\norm{l-s-r}_2^2+\frac\mu 2\norm{l}_2^2+\langle Wr,p_{n+1}\rangle + \langle \nabla r, q_{n+1}\rangle \\ + \langle \nabla l,u_{n+1}\rangle
+ \frac{1}{2\sigma}\norm{r-r_n}_2^2  +\frac{1}{2\sigma}\norm{l-l_n}_2^2 \Big\},
\end{multline}
\begin{equation} \label{updatebar}
\bar{r}_{n+1} = 2r_{n+1}-r_n,\quad \bar{l}_{n+1} = 2l_{n+1}-l_n.
\end{equation}
The optimization problems (\ref{updatep}--\ref{updaterl}) are all quadratic, so close form solutions can be easily obtained. For example, to solve the minimization problem \eqref{updaterl}, first notice that this problem can be separated into $MN$ single variable minimization problems
\begin{multline} \label{updaterl_single}
\min_{r(i)\geq 0,l(i)} E(r(i),l(i)) =  \Big\{ \frac\gamma 2(l(i)-s(i)-r(i))^2+\frac\mu 2 l(i)^2+a(i)r(i) + b(i)r(i) + c(i)l(i)
\\+ \frac{1}{2\sigma}(r(i)-r_n(i))^2 +\frac{1}{2\sigma}(l(i)-l_n(i))^2 \Big\},
\end{multline} where $a(i) = (W^Tp_{n+1})(i)$, $b(i) = (\nabla^Tq_{n+1})(i)$, and $c(i) = (\nabla^T u_{n+1})(i)$. Then the optimal condition of \eqref{updaterl_single} without the constraint $r(i)\geq 0$  leads to
\begin{equation} \label{Linear_System}
\begin{pmatrix} a_{11} & a_{12} \\ a_{21} & a_{22} \end{pmatrix} \begin{pmatrix} r(i) \\ l(i) \end{pmatrix}  = \begin{pmatrix} d_1 \\ d_2 \end{pmatrix},
\end{equation}where $a_{11} = \gamma+\frac 1\sigma$, $a_{12} = a_{21} = -\gamma $, $a_{22} = \gamma+\mu+\frac{1}{\sigma}$, $d_1 = \frac{1}{\sigma} r_n(i)-a(i)-b(i)-\gamma s(i)$ and $d_2 = \gamma s(i)-c(i)+\frac{1}{\sigma} l_n(i)$. Denote $(r_{n+\frac12}(i),l_{n+\frac12}(i))$ to be the solution to \eqref{Linear_System}, then it is clear that they are the solution to \eqref{updaterl_single} without the constraint. To enforce the constraint, we update as follows:
\begin{equation} \label{updaterl_explicit}
r_{n+1}(i) = \max\{ r_{n+\frac12}(i),0\}, \quad l_{n+1}(i) = \left\{ \begin{array}{rl} l_{n+\frac12}(i) & \text{if } r_{n+\frac12}(i)\geq0, \\ d_2/a_{22} & \text{otherwise}. \end{array} \right.
\end{equation}

\begin{proposition}
The $(r_{n+1}(i),l_{n+1}(i))$ defined in \eqref{updaterl_explicit} solves the constraint  minimization problem \eqref{updaterl_single}.
\end{proposition}
\begin{proof}
To show that $(r_{n+1}(i),l_{n+1}(i))$ solves \eqref{updaterl_single} with the constraint, first notice that \eqref{updaterl_single} is convex in both variables and $(r_{n+\frac12}(i),l_{n+\frac12}(i))$ solve the minimization problem \eqref{updaterl_single} without the constraint $r(i)\geq 0$. The case when $r_{n+\frac12}(i)\geq 0$ is clear, and now assume $r_{n+\frac12}(i)<0$. For any other point $(r(i),l(i))$ with $r(i)\geq 0$, denote $(0,l^*(i))$ to be the intersection point of the line segment between $(r(i),l(i))$ and $(r_{n+\frac12}(i), l_{n+\frac12}(i))$ and the line $r(i) = 0$. By the convexity of the objective functional $E$ in \eqref{updaterl_single} and the minimization property of $(r_{n+\frac12}(i), l_{n+\frac12}(i))$, we have $E(0,l^*(i))\leq E(r(i),l(i))$. Therefore, we have $E(r_{n+1}(i),l_{n+1}(i))\leq E(r(i),l(i))$, since in the case of $r_{n+\frac12}<0$, our definition of $l_{n+1}(i)$ minimizes the objective functional in \eqref{updaterl_single} on the line $r(i) = 0$.
\end{proof}

The following algorithm summarizes the procedures to solve the optimization problem~(\ref{mainmodel}).
	
\bigskip

\begin{center}
{\bf Algorithm 1:} Solving (\ref{mainmodel}) by the Chambolle-Pock algorithm \\
\vspace{0.1in}
\begin{tabular}{rll}\hline \\
 1.& Initialize: $p_0=0,q_0 = 0, u_0=0, r_0 = 0, l_0 = s,\bar{r}_0 = r_0, \bar{l}_0 = l_0$. \cr
 2.& Do $k=0, 1,\ldots,$ until convergence or reaches the maximum iteration number \\
  & (a) Update $p_{n+1}$ by projecting $\tau W\bar{r}_n+p_n$ to the boxed constraint $\abs{p}_2 \leq \abs{\diag G}_2$ \eqref{box_constraint} \cr
  & (b) Update $q_{n+1} = (\alpha \tau \nabla \bar{r}_n+\alpha q_n)/{(\tau+\sigma)}$. \cr
  & (c) Update $u_{n+1} = (\tau\beta \nabla \bar{l}_n + \beta u_n)/{(\tau+\beta)}$. \cr
  & (d) Update $r_{n+1}$ and $l_{n+1}$ by \eqref{updaterl_explicit}. \cr
  & (e) Update $\bar{r}_{n+1}$ and $\bar{l}_{n+1}$ by \eqref{updatebar}. \cr

 3.& Output: $r,l$. \cr \\
 \hline
\end{tabular}
\end{center}

\bigskip

In the following, we discuss the convergence of Algorithm 1.
\begin{proposition} \label{prop1}
The saddle point set of (\ref{pd_general}) is nonempty.
\end{proposition}

\noindent
The proof follows the same arguments as in Proposition~3.2 \cite{Box} (cf \cite{HL1993}).

Next we show the condition that guarantees the convergence of Algorithm 1.

\begin{proposition} \label{convergence}
Let $\norm{K}_2$ be the operator 2-norm of $K$ and $(x^{(n)}, y^{(n)})$ be defined by Algorithm~1. If we choose $\tau$ and $\sigma$ such that $\tau \sigma  <1/\norm{K}_2^2$, then $(x^{(n)},y^{(n)})$ converges to a saddle point $(x^*,y^*)$ of \eqref{pd_general}.
\end{proposition}

\noindent
The proposition is a special case of Theorem~1 in \cite{CP}. We remark that the limiting point $x^* = (r^*,l^*)$ is the unique solution pair of \eqref{mainmodel}. To see this, notice that \eqref{pd_general} is the primal-dual formulation of \eqref{mainmodel}. According to Proposition 3.1 of \cite{bookET}, if $(x^*,y^*)$ is a solution to \eqref{pd_general}, then $x^*$ is a solution to \eqref{mainmodel}. Since \eqref{mainmodel} has a unique minimizer, we conclude that $x^*$ is unique.

Lastly, we give an estimate for the bound of $\norm{K}_2$.

\begin{proposition} \label{operator_bound}
For the operator $K$ defined in \eqref{operator_def}, we have $\norm{K}_2 \leq 3$.
\end{proposition}
\begin{proof}
Since $W^TW = I$ by the ``perfect reconstruction property'', we have $\norm{W}_2 = 1$. It is known from \cite{ChamboJMIV} that $\norm{\nabla}_2^2 \leq 8$. Then we have
\begin{eqnarray}
\norm{K\begin{pmatrix} r \\ l\end{pmatrix}}_2 &=& \sqrt{\norm{Wr}_2^2+\norm{\nabla r}_2^2+\norm{\nabla l}_2^2}  \nonumber \\
                                              &\leq& \sqrt{(\norm{W}_2^2+\norm{\nabla}^2_2)\norm{r}_2^2+\norm{\nabla}_2^2\norm{l}_2^2} \nonumber \\
                                              &\leq& \sqrt{\norm{W}_2^2+\norm{\nabla}_2^2}\norm{\begin{pmatrix} r \\l \end{pmatrix}}_2 \nonumber \\
                                              &=& 3\norm{\begin{pmatrix} r \\l \end{pmatrix}}_2.
\end{eqnarray}
This shows that $\norm{K}_2 \leq 3$.
\end{proof}

\section{Numerical experiments} \label{sec:numerical_experiments}

In this section, we compare our tight-frame model \eqref{mainmodel} and our TV model \eqref{mainmodelTV} with some other segmentation methods in \cite{Li_MRI,CV,YBT,YBTBmul,Kaihua2016level}.
We report the results of the celebrated Chan-Vese segmentation method \cite{CV} as a baseline. All the methods \cite{Li_MRI,YBT,YBTBmul,Kaihua2016level} are efficient segmentation methods proposed after 2010. The method \cite{YBT} uses continuous max-flow and min-cut method to obtain two-phase segmentation, while the method \cite{YBTBmul} uses the same method to obtain multi-phase segmentation. We use these two methods to demonstrate segmentation results where intensity inhomogeneity is not considered explicitly. The methods in \cite{Li_MRI} and \cite{Kaihua2016level} are both designed to segment images with intensity inhomogeneity. The method \cite{Li_MRI} is a popular  multi-phase segmentation method (with more than 900 citations), while the method \cite{Kaihua2016level} is a more resent result (published in 2016) which can segment either 2 or 4 phase images. We compare with these methods to demonstrate the effectiveness of our model to handle intensity inhomogeneity. The parameters for different methods are chosen by trial and error to get the best results of the respective methods.

In our implementation of Algorithm 1, the parameters $\tau$ and $\sigma$ are fixed to $1$ and $0.1$ respectively. As explained previously, the parameter $\mu$ should be small, and it is fixed to $10^{-5}$ for all the experiments.  The parameters $\alpha,\beta$ and $\gamma$ need to be tuned for different images, and we list the values of them used in the tight frame regularisation in Table~\ref{table:list_parameters_TF}. For the weight matrix \eqref{weight_matrix}, we set $\epsilon = \frac{50}{MN}$, where $MN$ is the total number of pixels in a given image. To get a good implementation, we fix the iteration number of Algorithm 1 to $1000$. For the implementation of \eqref{mainmodelTV}, the parameters $\tau$, $\sigma$ and $\mu$ are fixed to $1$, $0.15$ and $10^{-5}$, and we terminate the iteration when $\norm{r_{n+1}-r_n}/\norm{r_n}\leq 10^{-5}$, or the maximum iteration number $1000$ is reached. In the second stage of our method, we set the thresholds $\rho$ manually to get good segmentation results. The values of the thresholds for the tight frame method are also included in Table~\ref{table:list_parameters_TF}.

\begin{table}[h!]
\centering
\begin{tabular}{|c|c|c|c|c|}
  \hline
  & $\alpha$ & $\beta$ & $\gamma$ & threshold\\  \hline
  Figure~\ref{Fish} & $10^{-3}$ & $80$ & $8$ & $0.9$ \\ \hline
  Figure~\ref{Boat} & $3$ & $45$ & $1.5$ & $0.89$ \\ \hline
  Figure~\ref{Feiji} & $1$ & $12$ & $5$ & $0.9$ \\ \hline
  Figure~\ref{Cells} & $5$ & $13$ & $1$ & $0.85$ \\ \hline
  Figure~\ref{Weakvessels} & $3*10^{-2}$ & $20$ & $3$ & $0.72$ \\ \hline
  Figure~\ref{Liver} & $5*10^{-2}$ & $15$ & $2$ & $0.35$ \\ \hline
  Figure~\ref{Animals} & $10^{-2}$ & $60$ & $5$ & $[0.55,0.75]$ \\ \hline
  Figure~\ref{Camel} & $5*10^{-2}$ & $15$ & $10$ & $[0.6,0.95]$ \\ \hline
  Figure~\ref{Brain} & $5*10^{-2}$ & $150$ & $50$ & $[0.4,0.6]$ \\ \hline

\end{tabular}
\caption{List of parameters for the tight frame method}
\label{table:list_parameters_TF}
\end{table}

\begin{figure*}[htbp]
\begin{center}
\begin{minipage}[t]{3 cm}
\includegraphics[height=2.3 cm]{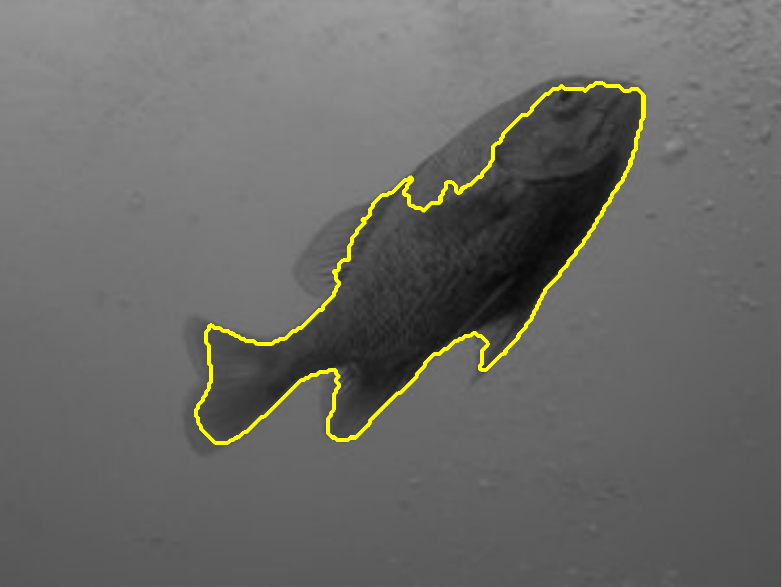}\\
\centering{(a) Chan--Vese \cite{CV}}
\end{minipage}
\begin{minipage}[t]{3 cm}
\includegraphics[height=2.3 cm]{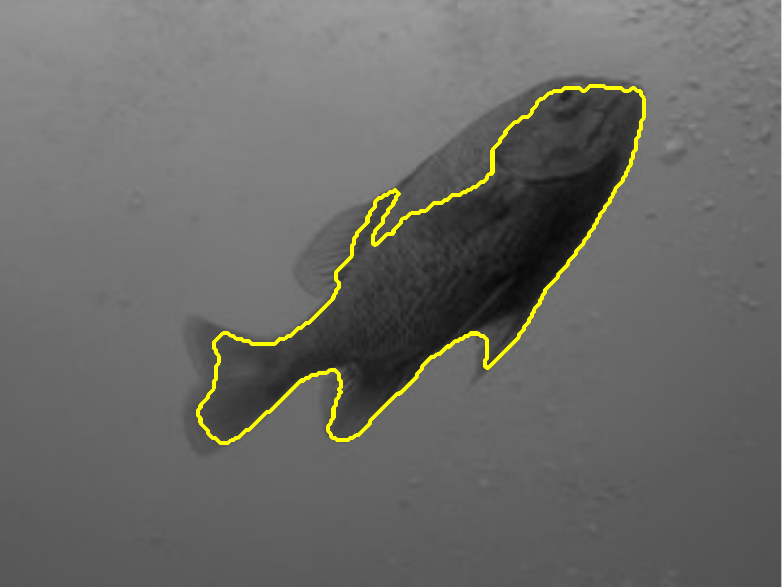}\\
\centering{(b) Yuan \cite{YBT}}
\end{minipage}
\begin{minipage}[t]{3 cm}
\includegraphics[height=2.3 cm]{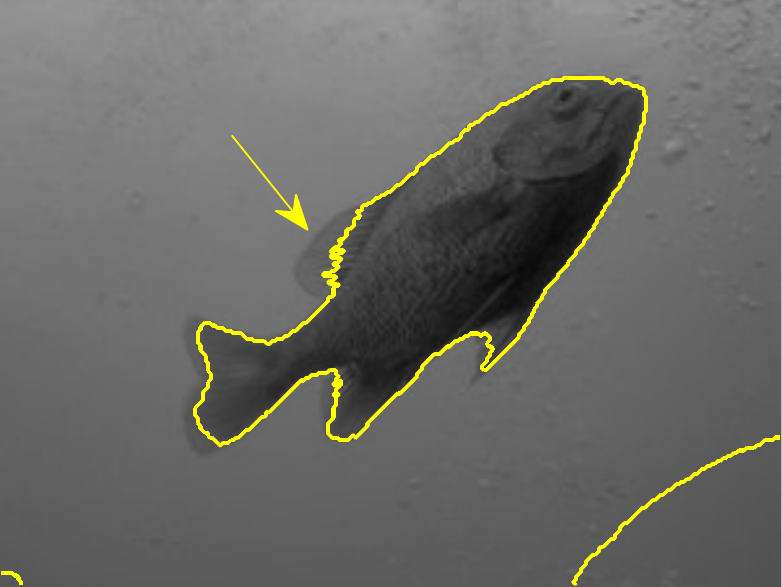}\\
\centering{(c) Li \cite{Li_MRI}}
\end{minipage}
\begin{minipage}[t]{3 cm}
\includegraphics[height=2.3 cm]{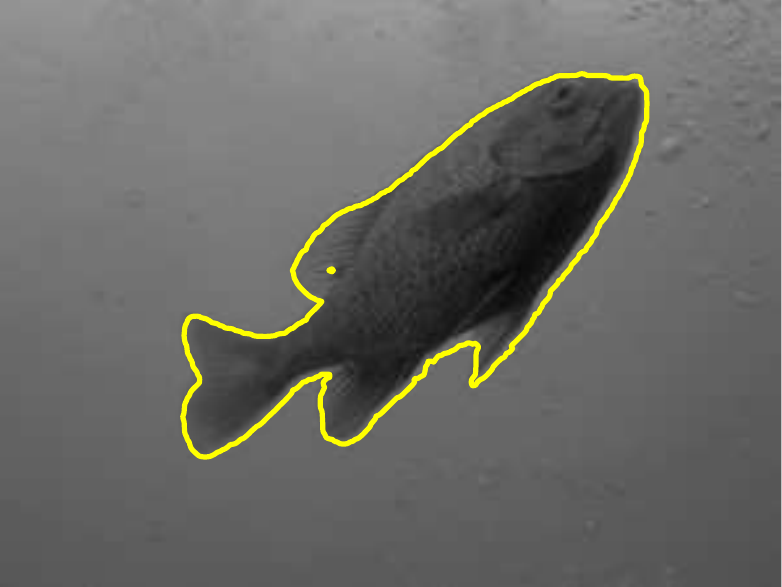}\\
\centering{(d) Zhang \cite{Kaihua2016level}}
\end{minipage}
\begin{minipage}[t]{3 cm}
\includegraphics[height=2.3 cm]{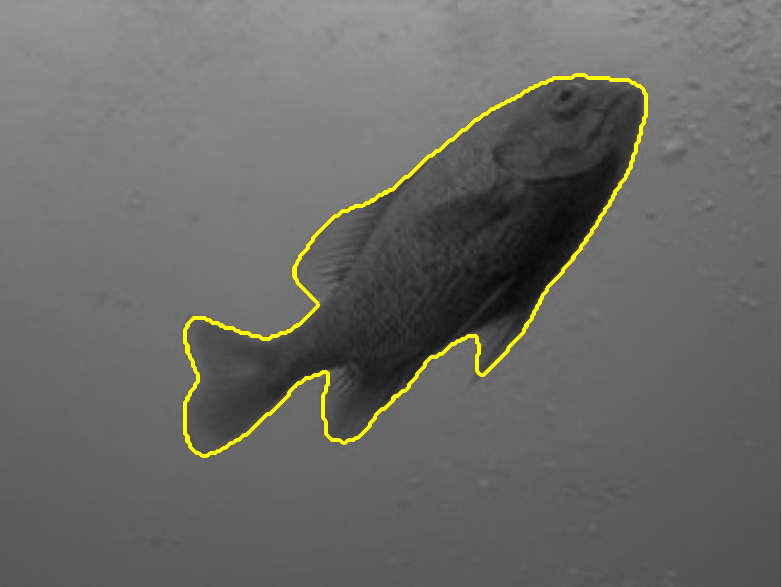}\\
\centering{(e) TV \eqref{mainmodelTV}}
\end{minipage}
\begin{minipage}[t]{3 cm}
\includegraphics[height=2.3 cm]{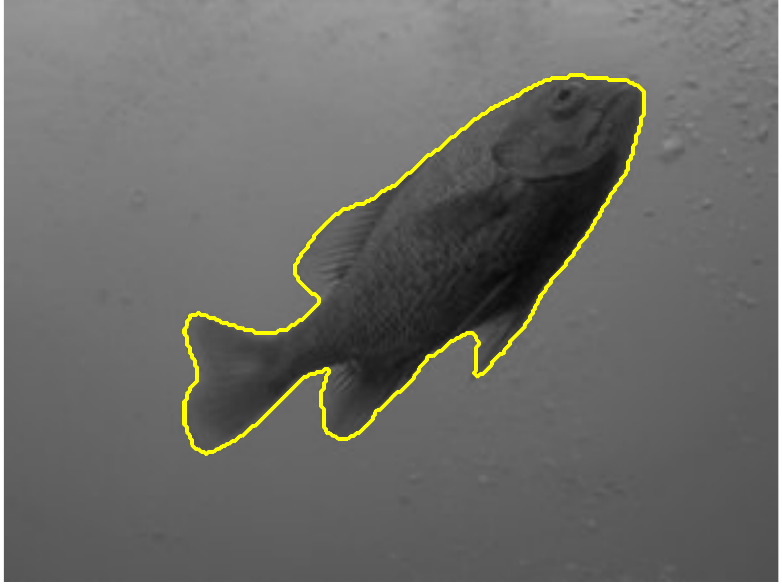}\\
\centering{(f) Tight-frame \eqref{mainmodel}}
\end{minipage}
\begin{minipage}[t]{3 cm}
\includegraphics[height=2.3 cm]{Fish_TFW_L-eps-converted-to.pdf}\\
\centering{(g) Illumination $L$ of Tight-frame}
\end{minipage}
\begin{minipage}[t]{3 cm}
\includegraphics[height=2.3 cm]{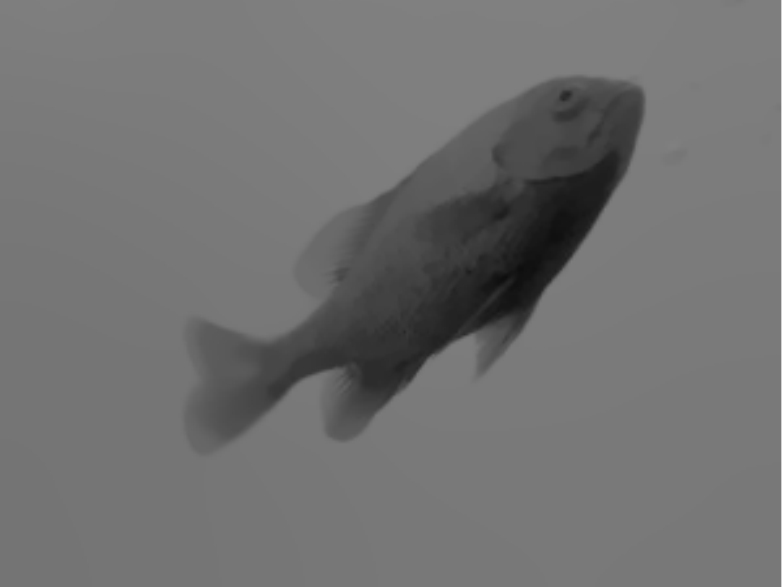}\\
\centering{(h) Reflection $R$ of Tight-frame}
\end{minipage}
\end{center}
\caption{\label{Fish}(a)--(d) results of \cite{CV}, \cite{YBT}, \cite{Li_MRI} and \cite{Kaihua2016level} respectively, (e) TV \eqref{mainmodelTV}, (f) Tight-frame \eqref{mainmodel},  (g) illumination by \eqref{mainmodel}, (h) reflection by \eqref{mainmodel}.}
\end{figure*}

\emph{Example \ref{Fish}}: Because of water and inhomogeneous light, this image is difficult to segment: the brightness of the water varies with the top being brighter, and the fish has both dark and bright parts. Both Fig. \ref{Fish}(a) from \cite{CV} and (b) from \cite{YBT} fail to segment the brighter part of the fish. Fig.~\ref{Fish}(c) from \cite{Li_MRI} segments the fish as a whole, but the segmentation lacks details (please refer to the arrow in (c)), and dark corners of the water are included in the segmentation. Fig.~\ref{Fish}(d) from \cite{Kaihua2016level} gives a successful segmentation with fine details. At the same time, no dark part of the water is included.  Fig. \ref{Fish}(e) and (f) from our methods both get successful segmentations, with no corner of lower intensity included. Fig.~\ref{Fish}(g) is the illumination part from the tight-frame regularisation \eqref{mainmodel}. We see that this image is very smooth and upper part of the water is brighter. Fig.~\ref{Fish}(h) is the reflection from the tight-frame regularisation \eqref{mainmodel}. It is clear that the fish stands out in this image, with the background flattened. Meanwhile the fins of the fish, which is weak in the original image, is well preserved in the reflection image, and this facilitates a good segmentation of the fish.

\begin{figure*}[htbp]
\begin{center}
\begin{minipage}[t]{4.5 cm}
\includegraphics[height=3.35 cm]{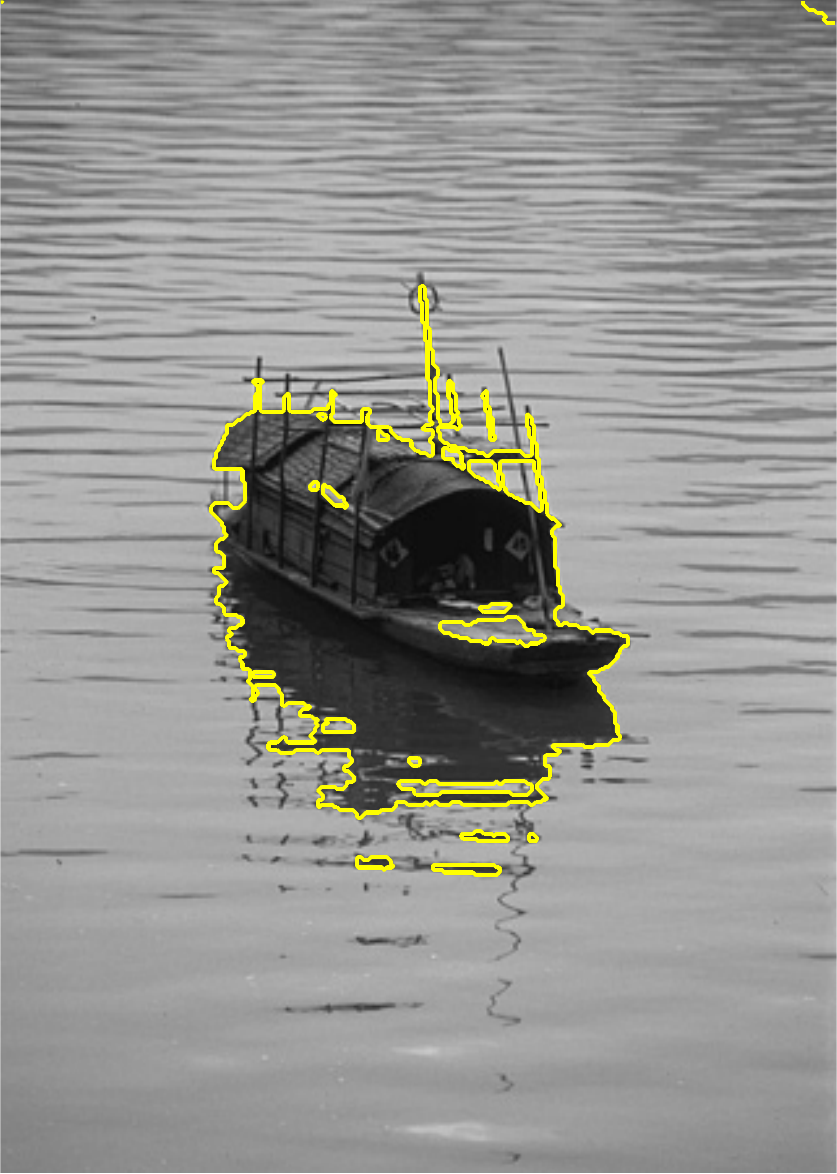}\\
\centering{(a) Chan--Vese \cite{CV}}
\end{minipage}
\begin{minipage}[t]{4.5 cm}
\includegraphics[height=3.35 cm]{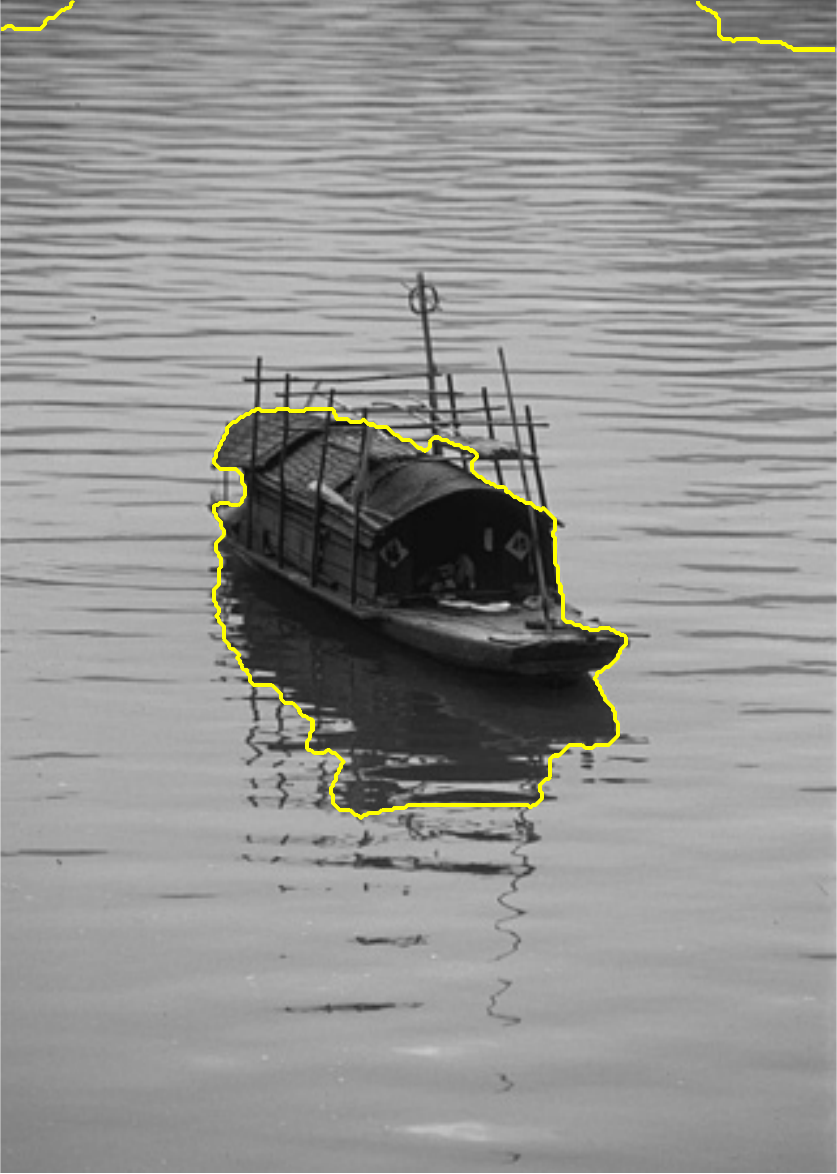}\\
\centering{(b) Yuan \cite{YBT}}
\end{minipage}
\begin{minipage}[t]{4.5 cm}
\includegraphics[height=3.35 cm]{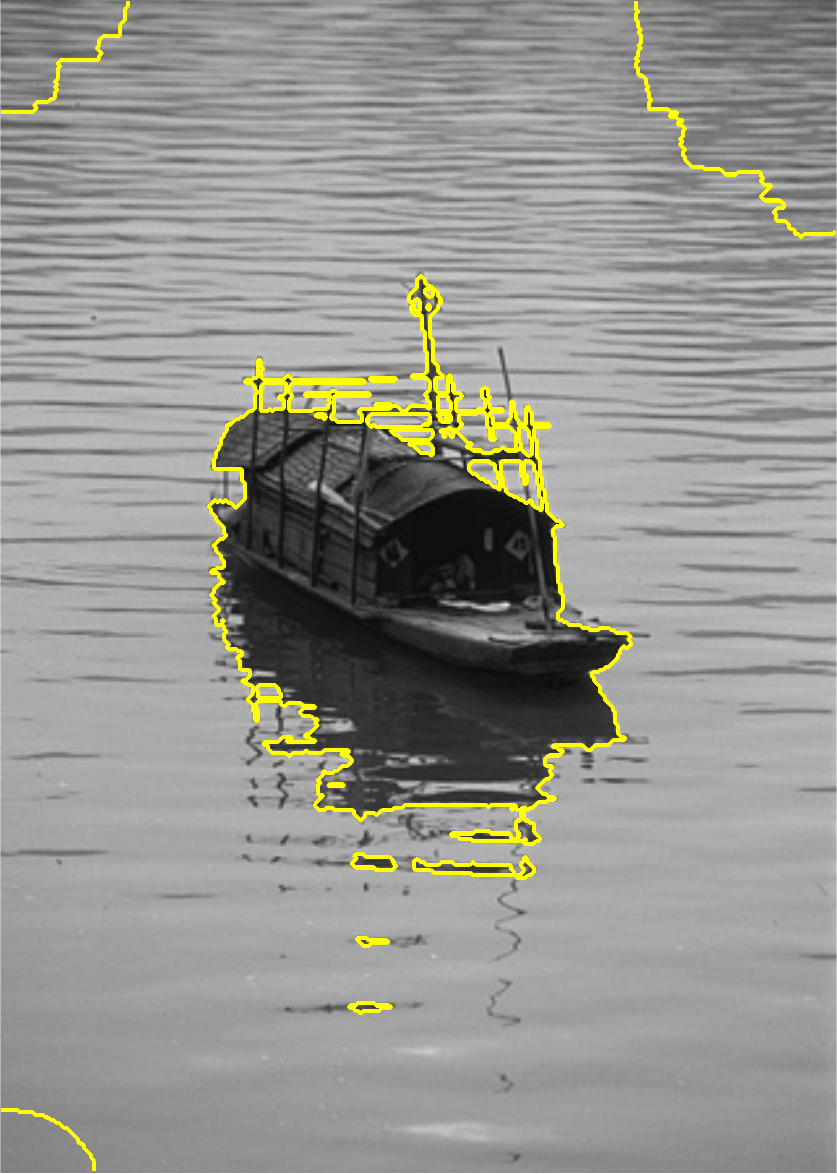}\\
\centering{(c) Li \cite{Li_MRI}}
\end{minipage}
\begin{minipage}[t]{4.5 cm}
\includegraphics[height=3.35 cm]{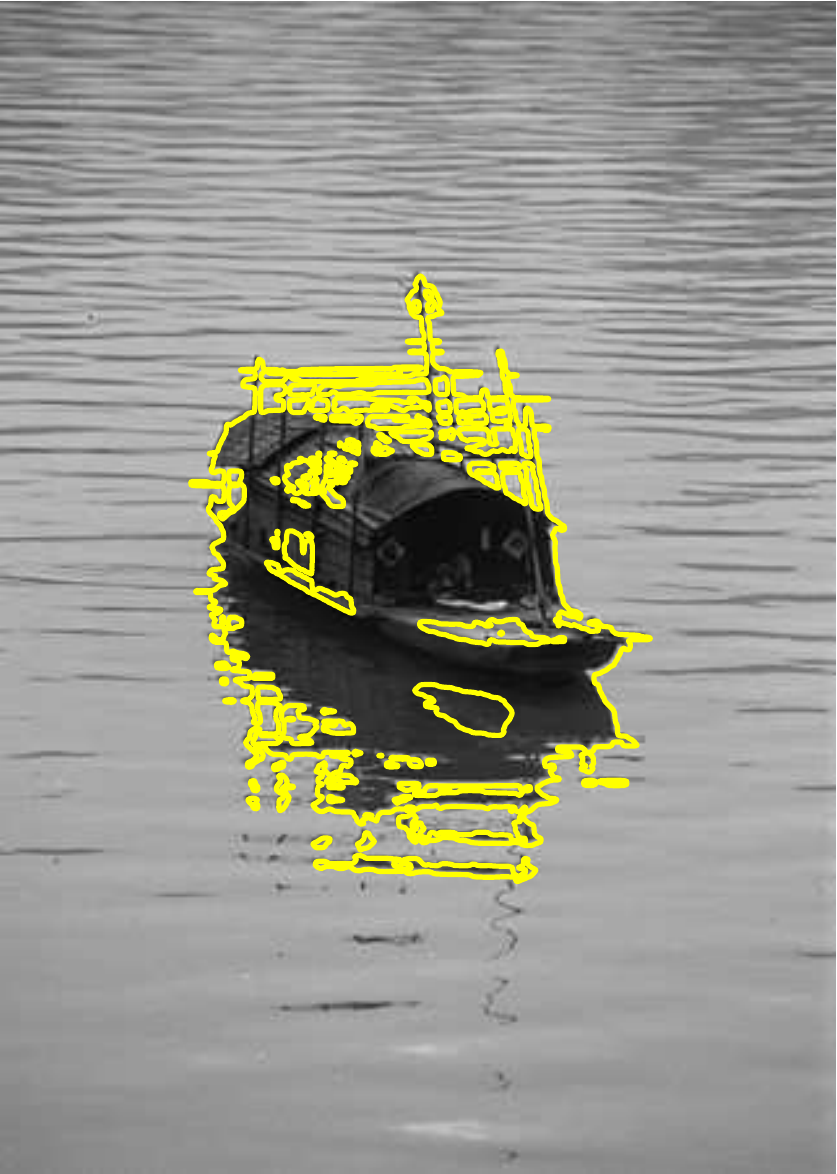}\\
\centering{(d) Zhang \cite{Kaihua2016level}}
\end{minipage}
\begin{minipage}[t]{4.5 cm}
\includegraphics[height=3.35 cm]{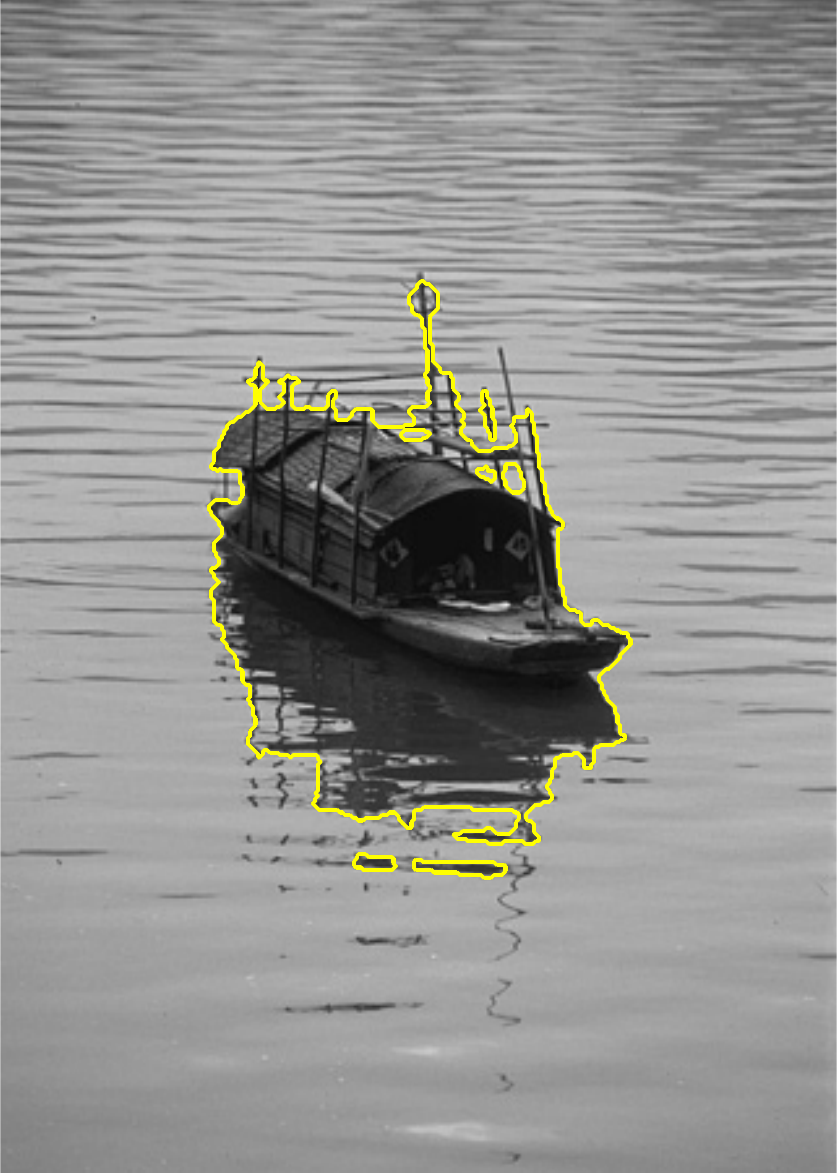}\\
\centering{(e) TV \eqref{mainmodelTV}}
\end{minipage}
\begin{minipage}[t]{4.5 cm}
\includegraphics[height=3.35 cm]{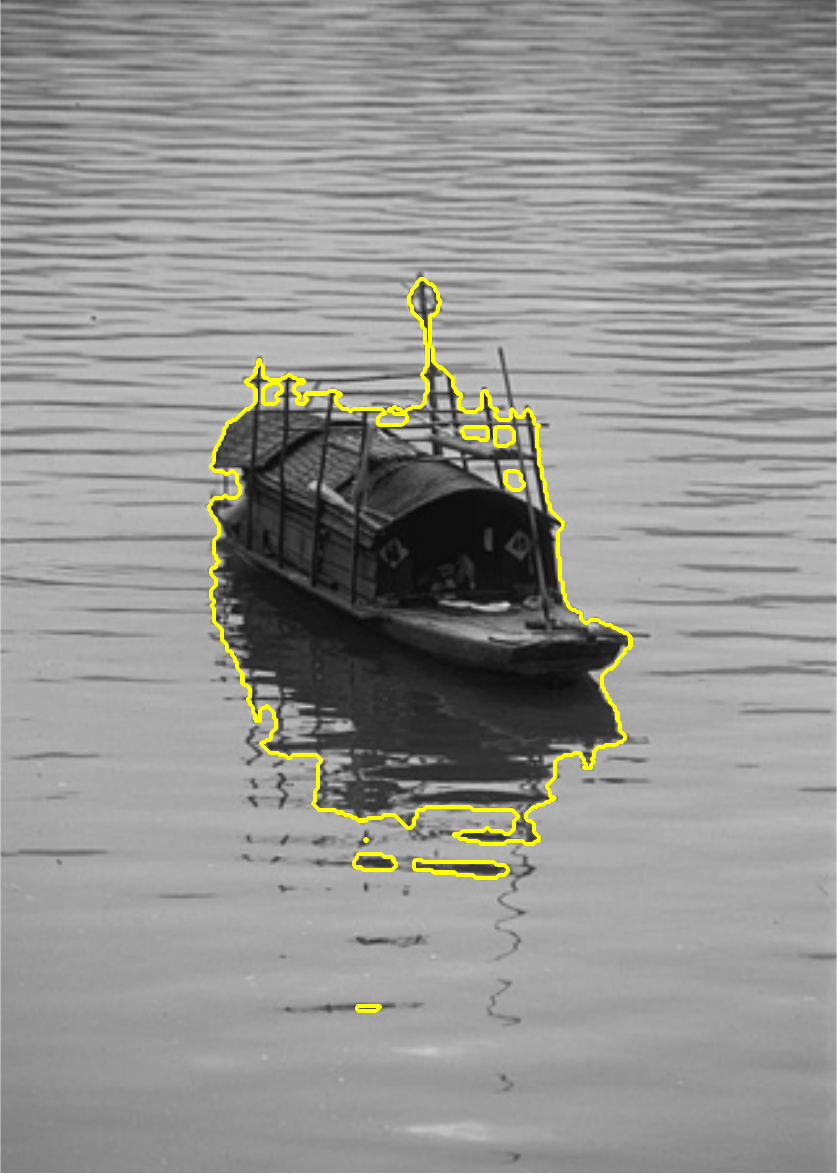}\\
\centering{(f) Tight-frame \eqref{mainmodel}}
\end{minipage}
\end{center}
\caption{\label{Boat}(a)--(d) results of \cite{CV}, \cite{YBT}, \cite{Li_MRI} and \cite{Kaihua2016level} respectively, (e) TV \eqref{mainmodelTV}, (f) Tight-frame \eqref{mainmodel}. }
\end{figure*}

\emph{Example \ref{Boat}}: This image is from the Berkeley Segmentation Dataset and Benchmark \cite{BerkleyDataset}. It is challenging to segment this image for two reasons: the corners of this image have darker intensities, and the boat has inhomogeneous intensity. The goal of this segmentation is to segment the boat and its reflection in the water as a single object, without including pixels from the four corners of the image. Fig. \ref{Boat}(a) from \cite{CV} fails to segment the boat as a whole, while (b) from \cite{YBT} includes the upper corners of the image in the segmentation. Fig.~\ref{Boat}(c) from \cite{Li_MRI} gives a segmentation with fine details, but it included three corners of the water in the segmentation. Fig.~\ref{Boat}(d) from \cite{Kaihua2016level} fails to segment the boat as a whole. Fig. \ref{Boat}(e) and (f) from our methods both get successful segmentations, with no corner of lower intensity included.

\begin{figure*}[htbp]
\begin{center}
\begin{minipage}[t]{4 cm}
\includegraphics[height=2.5 cm]{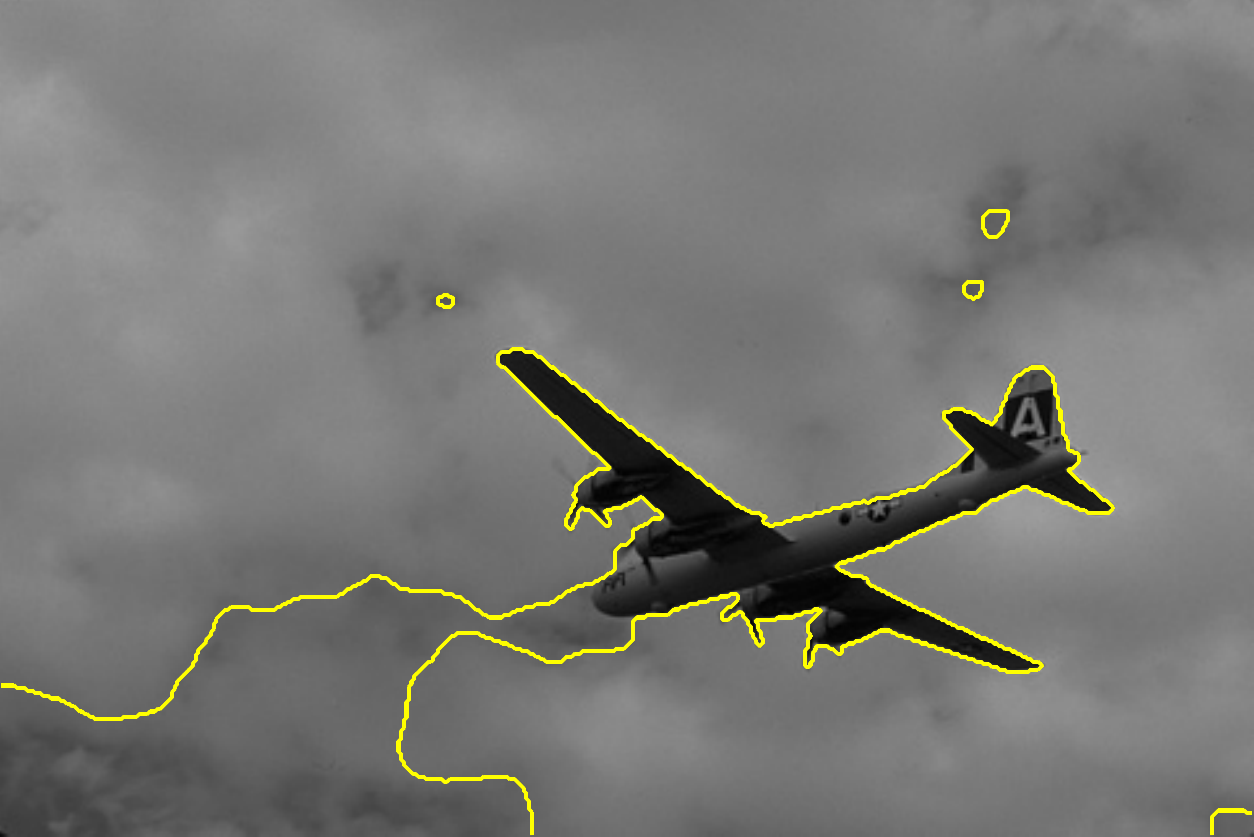}\\
\centering{(a) Chan--Vese \cite{CV}}
\end{minipage}
\begin{minipage}[t]{4 cm}
\includegraphics[height=2.5 cm]{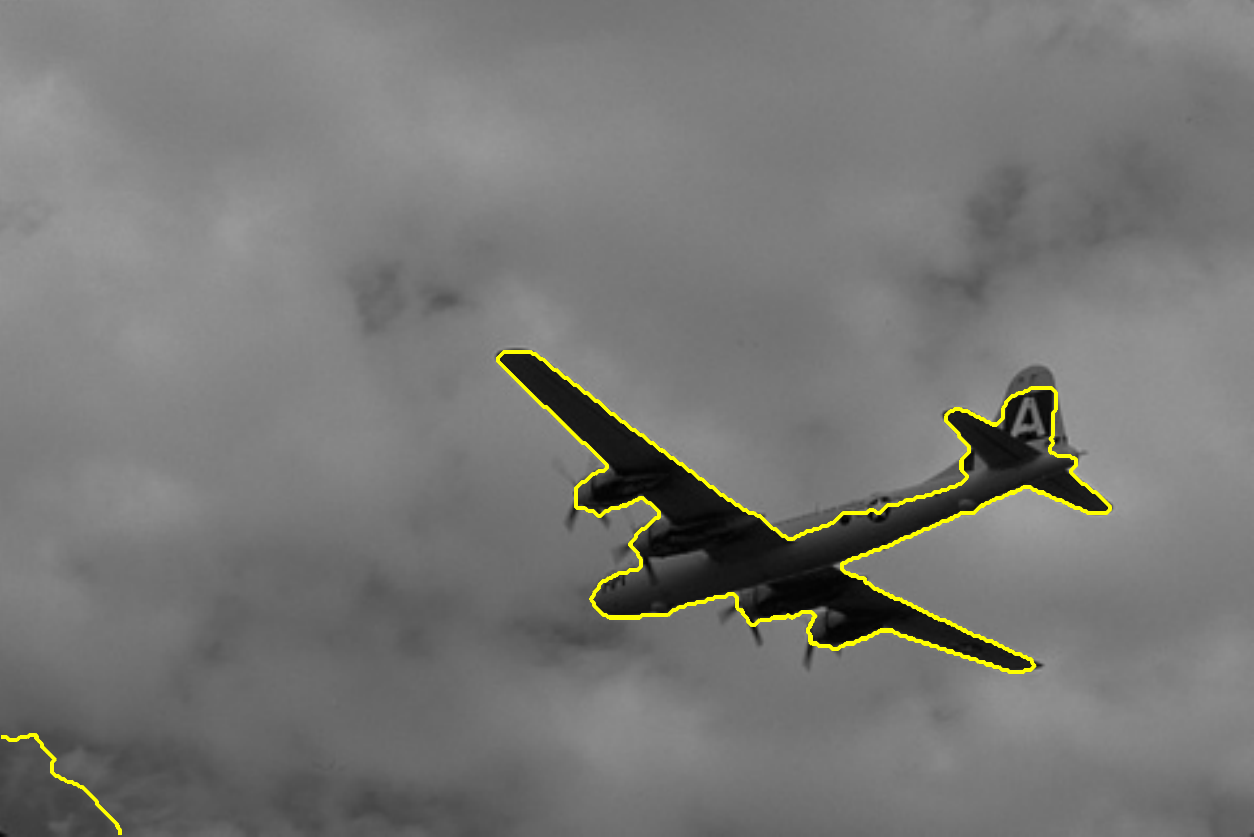}\\
\centering{(b) Yuan \cite{YBT}}
\end{minipage}
\begin{minipage}[t]{4 cm}
\includegraphics[height=2.5 cm]{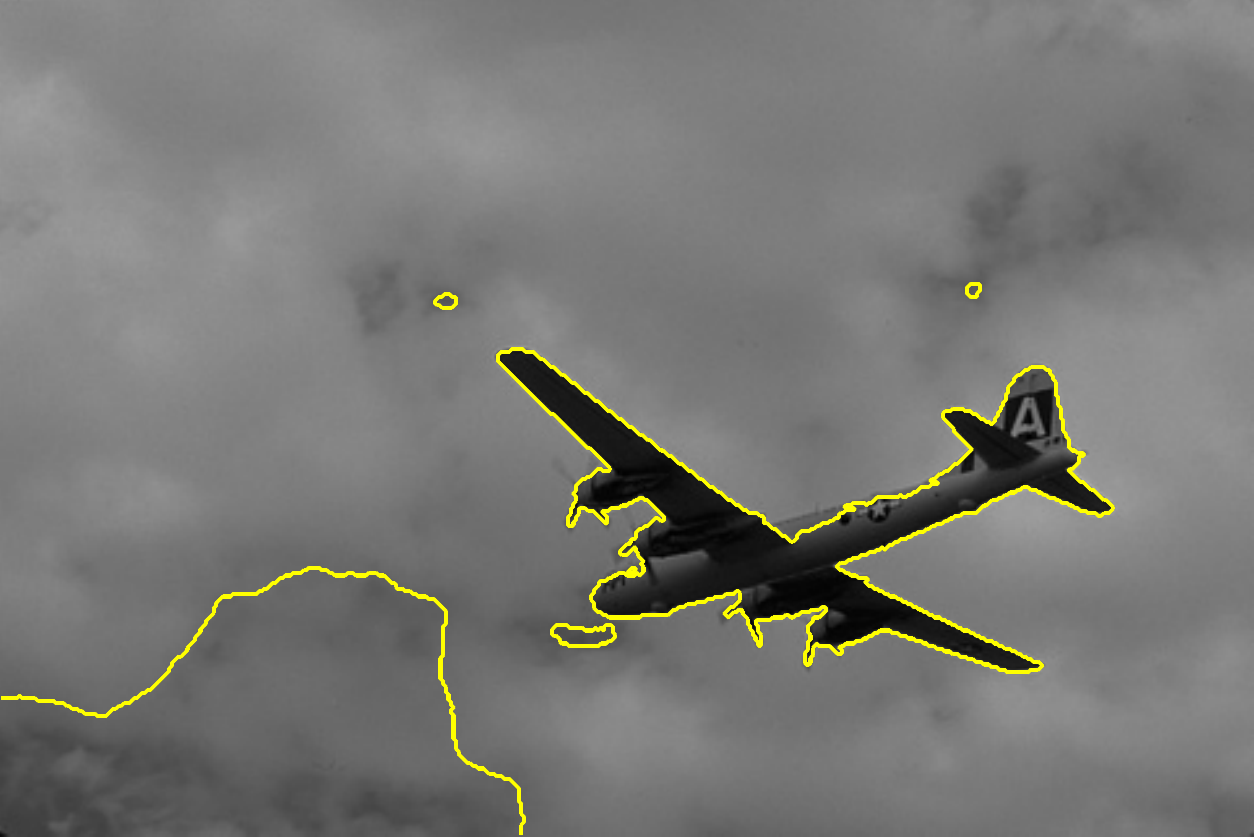}\\
\centering{(c) Li \cite{Li_MRI}}
\end{minipage}
\begin{minipage}[t]{4 cm}
\includegraphics[height=2.5 cm]{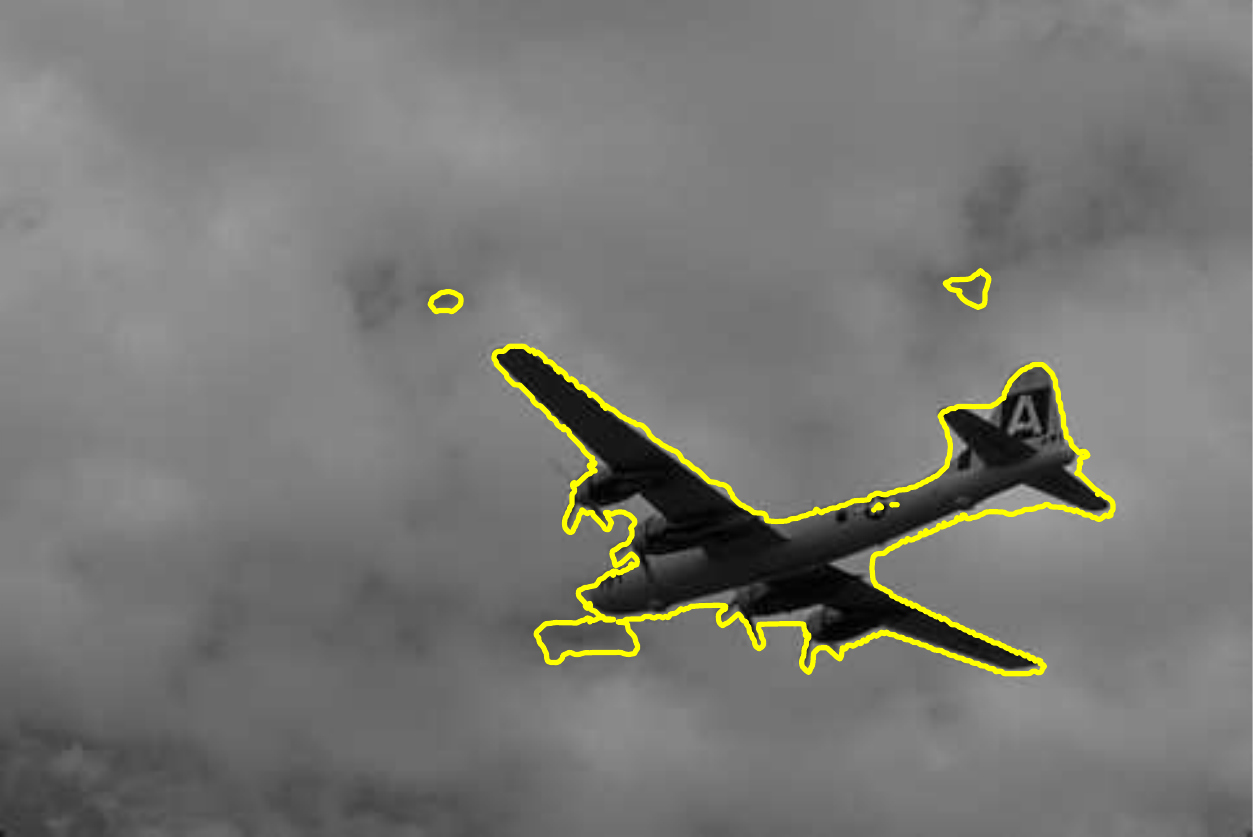}\\
\centering{(d) Zhang \cite{Kaihua2016level}}
\end{minipage}
\begin{minipage}[t]{4 cm}
\includegraphics[height=2.5 cm]{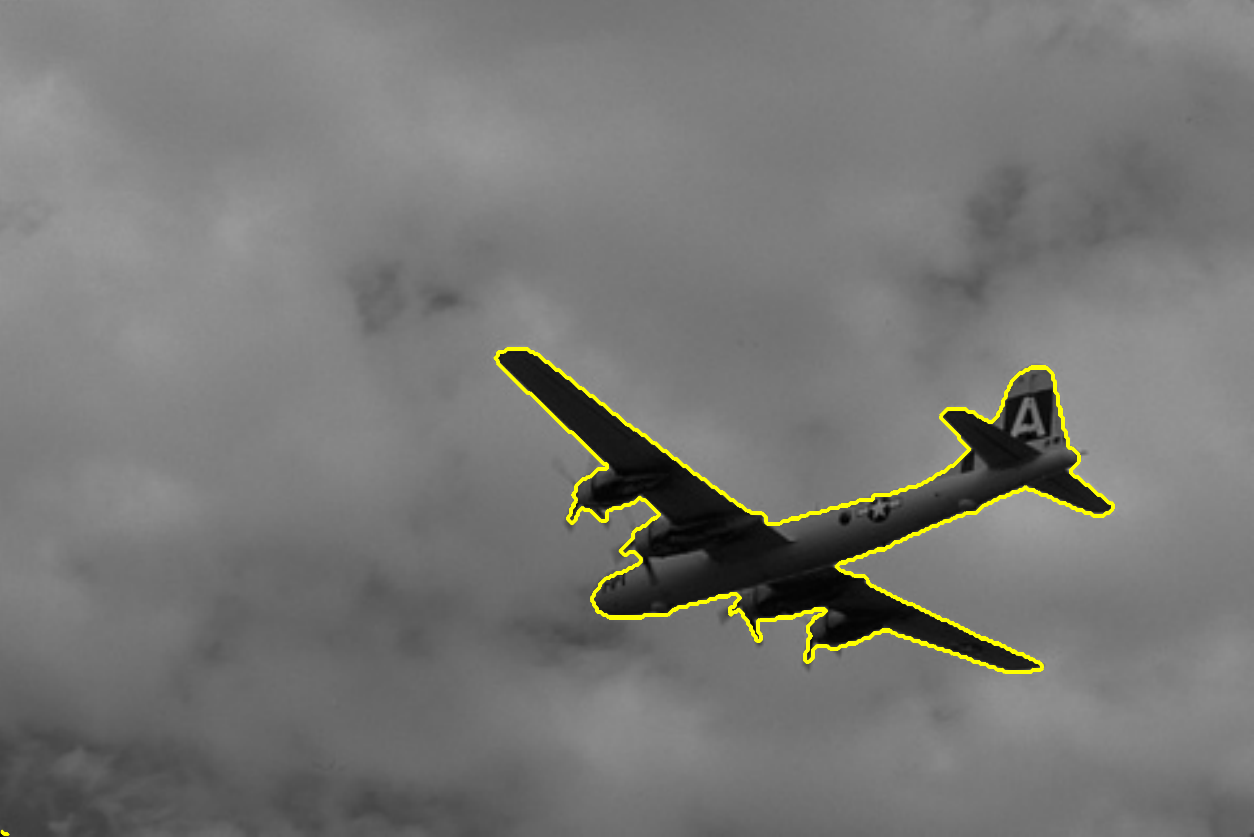}\\
\centering{(e) TV \eqref{mainmodelTV}}
\end{minipage}
\begin{minipage}[t]{4 cm}
\includegraphics[height=2.5 cm]{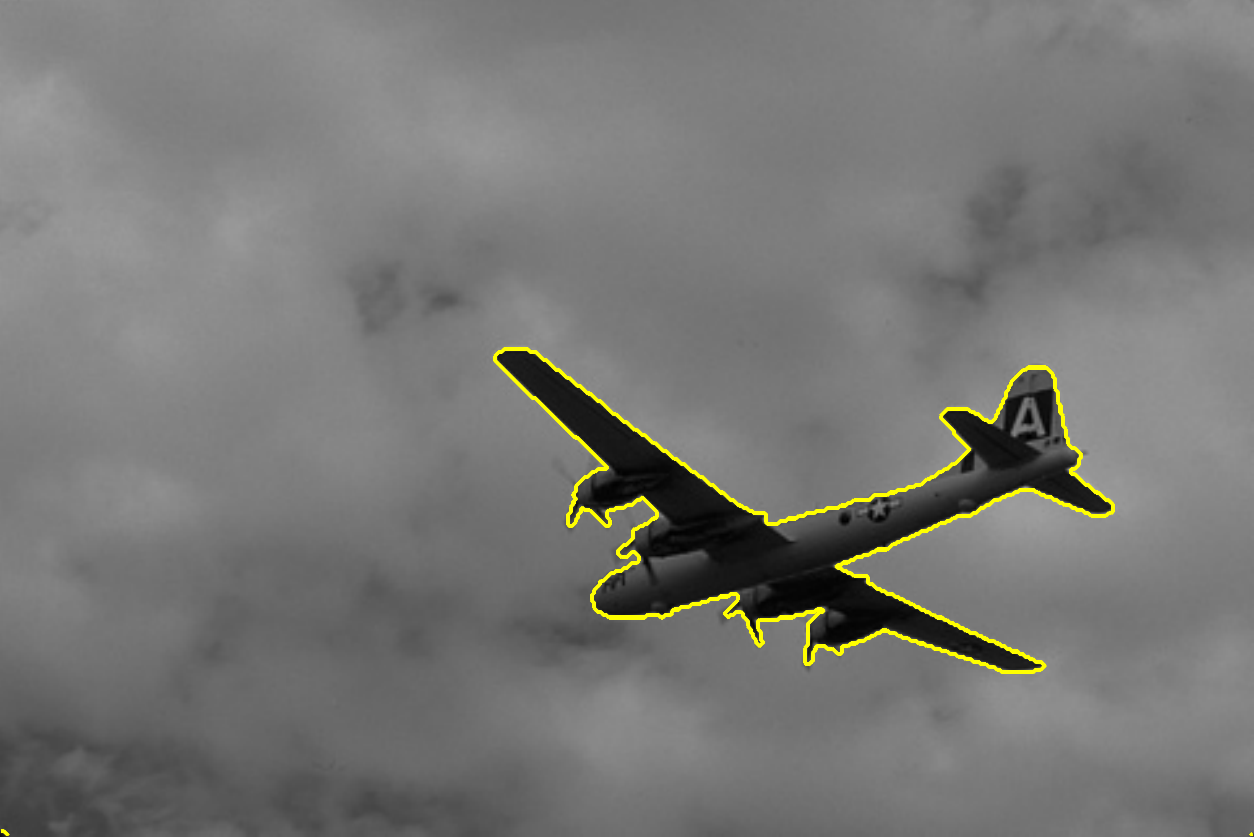}\\
\centering{(f) Tight-frame \eqref{mainmodel}}
\end{minipage}
\begin{minipage}[t]{3.5 cm}
\includegraphics[height=2.3 cm]{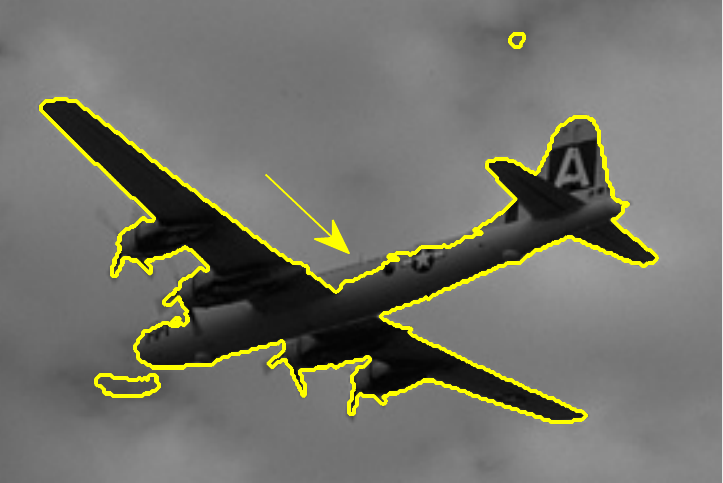}\\
\centering{(g) Detail of (c)}
\end{minipage}
\begin{minipage}[t]{3.5 cm}
\includegraphics[height=2.3 cm]{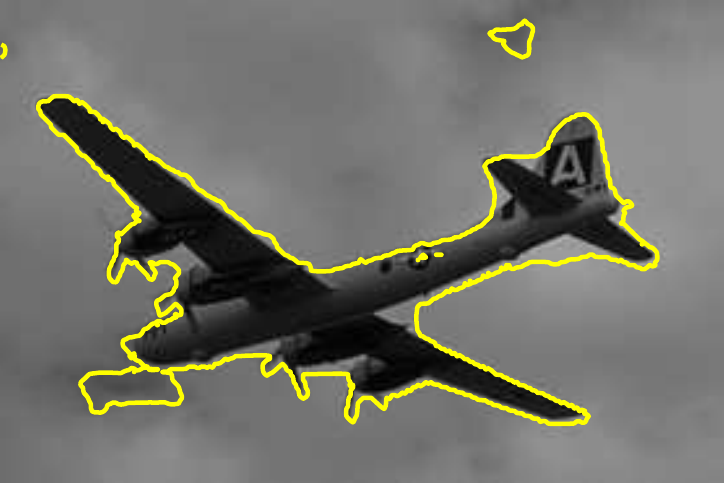}\\
\centering{(h) Detail of (d)}
\end{minipage}
\begin{minipage}[t]{3.5 cm}
\includegraphics[height=2.3 cm]{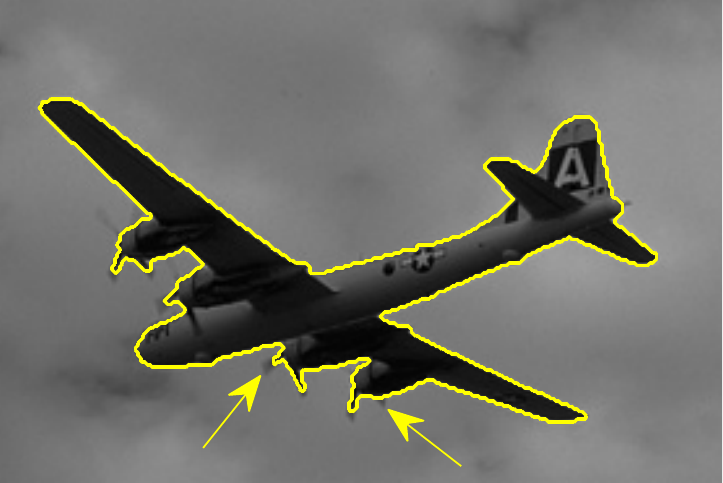}\\
\centering{(i) Detail of (e)}
\end{minipage}
\begin{minipage}[t]{3.5 cm}
\includegraphics[height=2.3 cm]{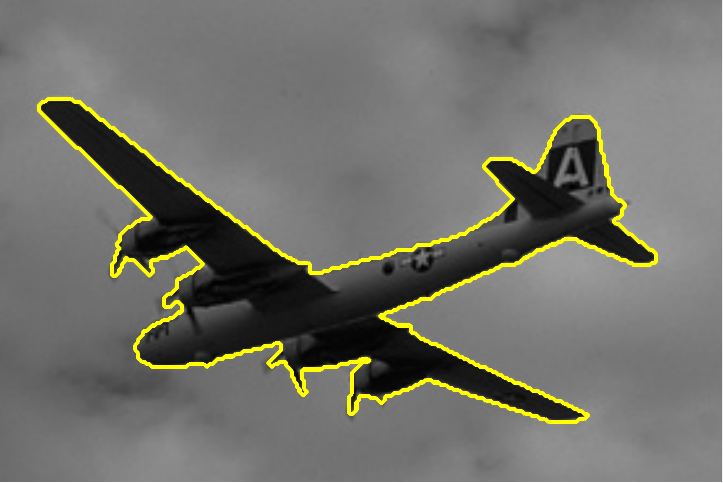}\\
\centering{(j) Detail of (f)}
\end{minipage}
\end{center}
\caption{\label{Feiji}(a)--(d) results of \cite{CV}, \cite{YBT}, \cite{Li_MRI} and \cite{Kaihua2016level} respectively, (e) TV \eqref{mainmodelTV}, (f) Tight-frame \eqref{mainmodel}, (g)--(j) details of (c)--(f). }
\end{figure*}

\emph{Example \ref{Feiji}}: This image is from the Berkeley Segmentation Dataset and Benchmark \cite{BerkleyDataset}. The varying intensities of the cloud, as well as some very weak boundaries of the aircraft, makes this segmentation very challenging. Fig.~\ref{Feiji}(a) from \cite{CV}  includes a large part of the lower left corner of the background in the segmention. Fig.~\ref{Feiji}(b) from \cite{YBT} segments the aircraft as a whole, but the boundary of the aircraft is not well detected, and the lower left corner is included in the segmentation. Fig.~\ref{Feiji}(c) from \cite{Li_MRI} segments the aircraft as a whole, but it also includes a large part of the lower left corner in the segmentation, and the boundary of the aircraft is not well detected. Fig.~\ref{Feiji}(d) from \cite{Kaihua2016level} manages to segment the aircraft as a single object, but the boundary of the aircraft is not well detected. It is clear that both Fig.~\ref{Feiji}(e) and (f) from our methods give very good segmentations, with the aircraft segmented with fine details. Fig.~\ref{Feiji}(g)--(j) are the details of segmentations in (c)--(f), where the arrow in (g) indicates a wrongly detected boundary, while the arrows in (i) indicate that the propellers of the aircraft are not well segmented (compared to (j)).

\begin{figure*}[htbp]
\begin{center}
\begin{minipage}[t]{4 cm}
\includegraphics[height=3 cm]{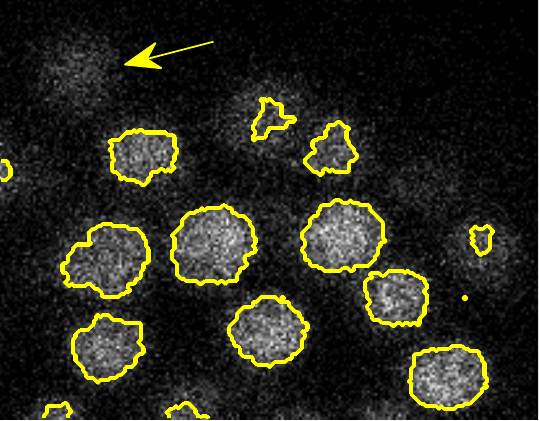}\\
\centering{(a) Chan--Vese \cite{CV}}
\end{minipage}
\begin{minipage}[t]{4 cm}
\includegraphics[height=3 cm]{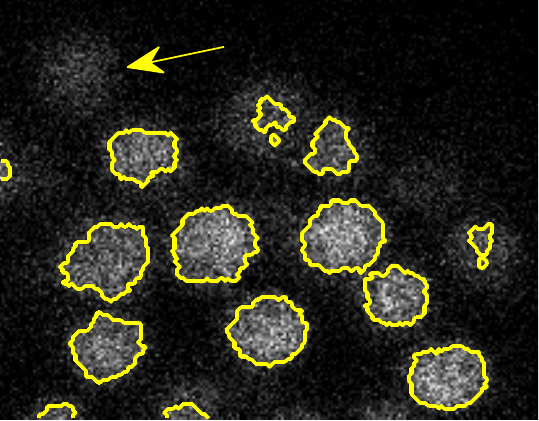}\\
\centering{(b) Yuan \cite{YBT}}
\end{minipage}
\begin{minipage}[t]{4 cm}
\includegraphics[height=3 cm]{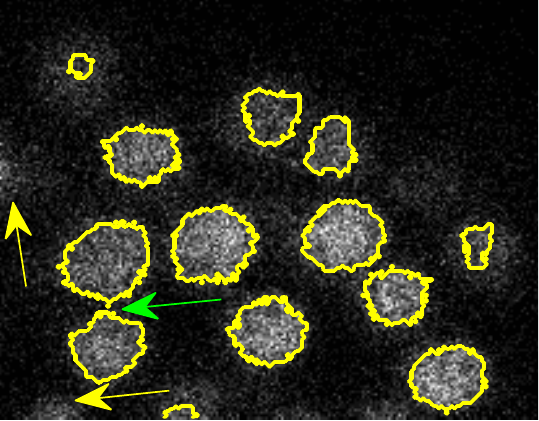}\\
\centering{(c) Li \cite{Li_MRI}}
\end{minipage}
\begin{minipage}[t]{4 cm}
\includegraphics[height=3 cm]{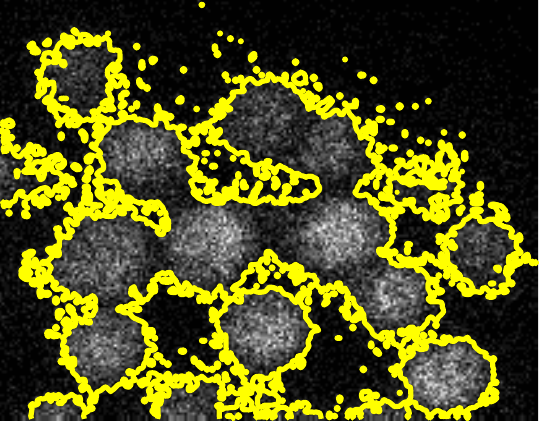}\\
\centering{(d) Zhang \cite{Kaihua2016level}}
\end{minipage}
\begin{minipage}[t]{4 cm}
\includegraphics[height=3 cm]{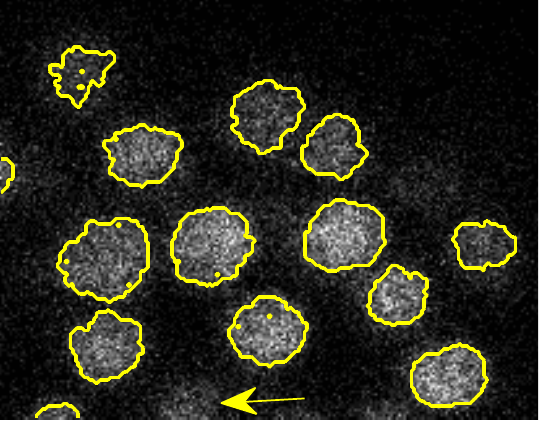}\\
\centering{(e) TV \eqref{mainmodelTV}}
\end{minipage}
\begin{minipage}[t]{4 cm}
\includegraphics[height=3 cm]{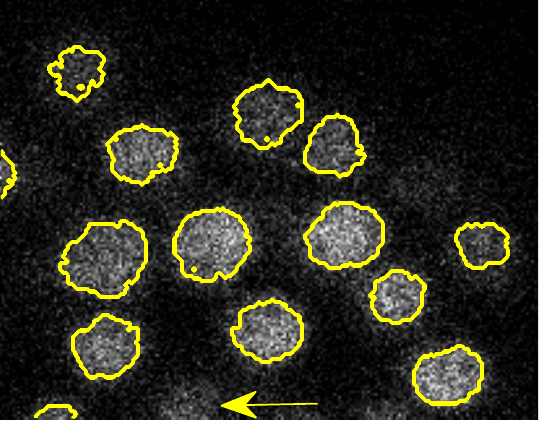}\\
\centering{(f) Tight-frame \eqref{mainmodel}}
\end{minipage}
\end{center}
\caption{\label{Cells}(a)--(d) results of \cite{CV}, \cite{YBT}, \cite{Li_MRI} and \cite{Kaihua2016level} respectively, (e) TV \eqref{mainmodelTV}, (f) Tight-frame \eqref{mainmodel}. }
\end{figure*}

\emph{Example \ref{Cells}}: Fig. \ref{Cells} is extracted from  a noisy real image from an automated cell tracking system \cite{AutomatedCell}, where the authors developed a system to track cell lineage during  Caenorhabditis elegans embryogenesis under low exposure of lights. In their experiments, noise in the images led to false positives in nuclear identification. This image is difficult to segment because of high level of noise and intensity inhomogeneity: it can be seen that some cells have high intensities while other cells have lower intensities. Our goal in this experiment is to segmented all the isolated cells. Fig.~\ref{Cells}(a) segments most cells, with one obvious cell excluded (see the arrow in (a)). Fig.~\ref{Cells}(b) from \cite{YBT}  produces a result similar to (a). Fig.~\ref{Cells}(c) from \cite{Li_MRI} produces unnatural boundaries of the cells, e.g. see the green arrow in (c), and two obvious cells are excluded, see the yellow arrows. Fig.~\ref{Cells}(d) from \cite{Kaihua2016level} fails to produce a good segmentation. Fig.~\ref{Cells}(e) and (f) from our methods produce good results, where the cells are well separated with smooth boundaries, but with one obvious cell excluded in the segmentation: see the arrows in (e) and (f).

\begin{figure*}[htbp]
\begin{center}
\begin{minipage}[t]{4.2 cm}
\includegraphics[height=4 cm]{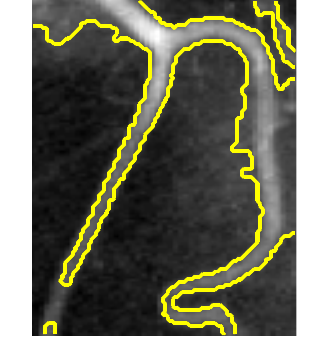}\\
\centering{(a) Chan--Vese \cite{CV}}
\end{minipage}
\begin{minipage}[t]{4.2 cm}
\includegraphics[height=4 cm]{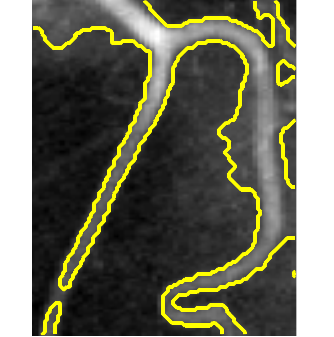}\\
\centering{(b) Yuan \cite{YBT}}
\end{minipage}
\begin{minipage}[t]{4.2 cm}
\includegraphics[height=4 cm]{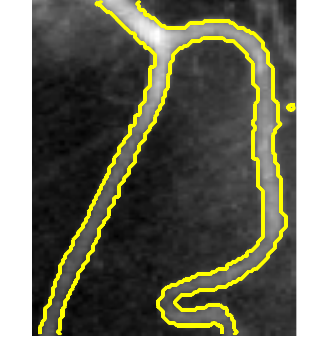}\\
\centering{(c) Li \cite{Li_MRI}}
\end{minipage}
\begin{minipage}[t]{4.2 cm}
\includegraphics[height=4 cm]{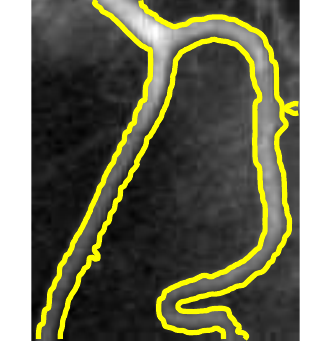}\\
\centering{(d) Zhang \cite{Kaihua2016level}}
\end{minipage}
\begin{minipage}[t]{4.2 cm}
\includegraphics[height=4 cm]{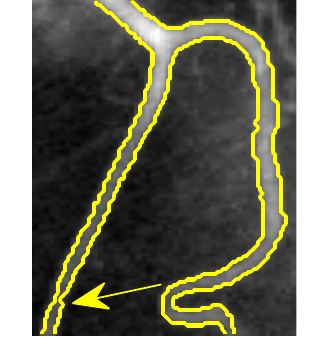}\\
\centering{(e) TV \eqref{mainmodelTV}}
\end{minipage}
\begin{minipage}[t]{4.2 cm}
\includegraphics[height=4 cm]{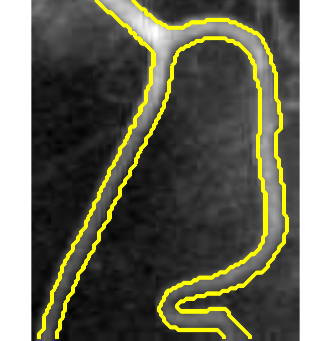}\\
\centering{(f) Tight-frame \eqref{mainmodel}}
\end{minipage}
\end{center}
\caption{\label{Weakvessels}(a)--(d) results of \cite{CV}, \cite{YBT}, \cite{Li_MRI} and \cite{Kaihua2016level} respectively, (e) TV \eqref{mainmodelTV}, (f) Tight-frame \eqref{mainmodel}. }
\end{figure*}

\emph{Example \ref{Weakvessels}}: Fig. \ref{Weakvessels} is an image of blood vessels. Notice that the top left corner of the image has a higher intensity than the left branch of the blood vessel, which makes the segmentation challenging. Fig.~\ref{Weakvessels}(a) from \cite{CV} fails to give a proper segmentation of the vessels: part of the upper left corner is included in the segmentation while the left branch of the vessel is disconnected. Fig.~\ref{Weakvessels}(b) from \cite{YBT} again produces a similar result as (a). Both Fig.~\ref{Weakvessels}(c) from \cite{Li_MRI} and (d) from \cite{Kaihua2016level} give satisfactory segmentations. Fig. \ref{Weakvessels}(e) from the TV regularisation manages to segment the vessels as a whole, but part of the left branch of the vessels is too narrow, see the arrow in (e). Fig. \ref{Weakvessels}(f) from the tight-frame regularisation produces a satisfactory segmentation.

\begin{figure*}[htbp]
\begin{center}
\begin{minipage}[t]{4.5 cm}
\includegraphics[height=2.81 cm]{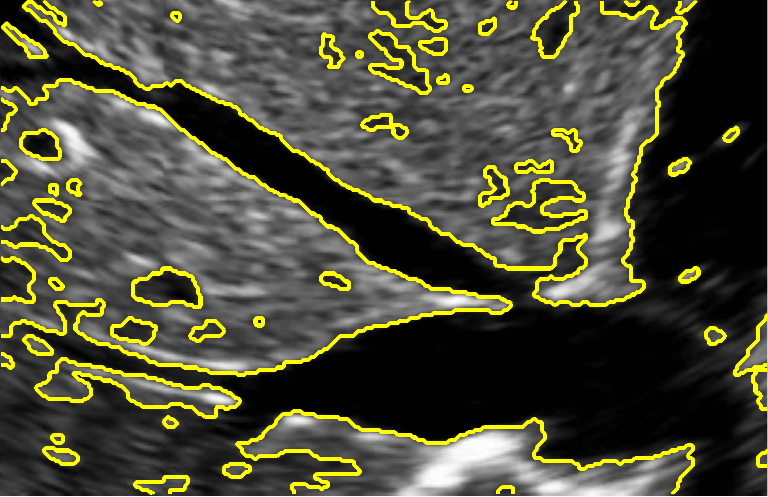}\\
\centering{(a)  Chan--Vese \cite{CV}}
\end{minipage}
\begin{minipage}[t]{4.5 cm}
\includegraphics[height=2.81 cm]{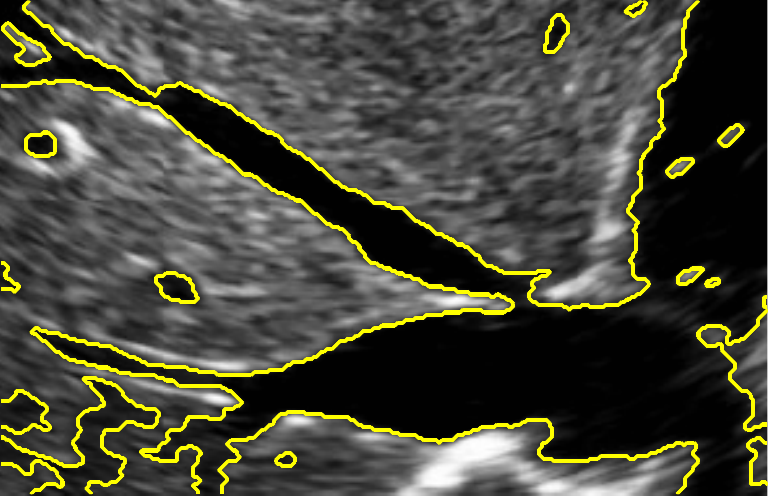}\\
\centering{(b) Yuan \cite{YBT}}
\end{minipage}
\begin{minipage}[t]{4.5 cm}
\includegraphics[height=2.81 cm]{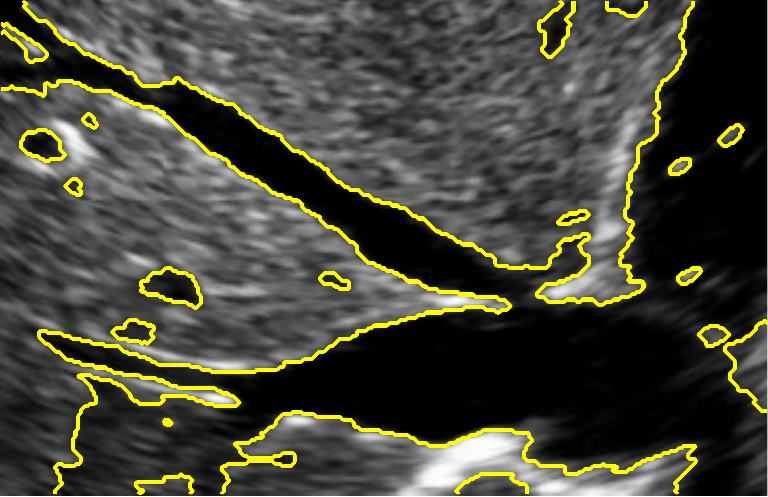}\\
\centering{(c) Li \cite{Li_MRI}}
\end{minipage}
\begin{minipage}[t]{4.5 cm}
\includegraphics[height=2.81 cm]{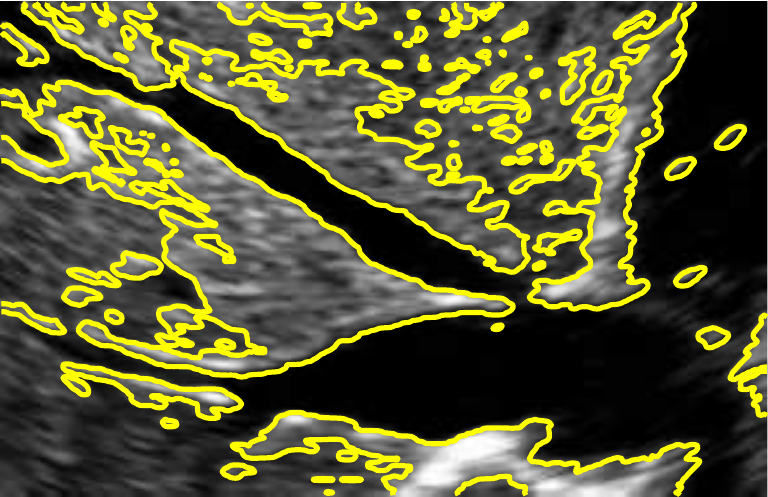}\\
\centering{(d) Zhang \cite{Kaihua2016level}}
\end{minipage}
\begin{minipage}[t]{4.5 cm}
\includegraphics[height=2.81 cm]{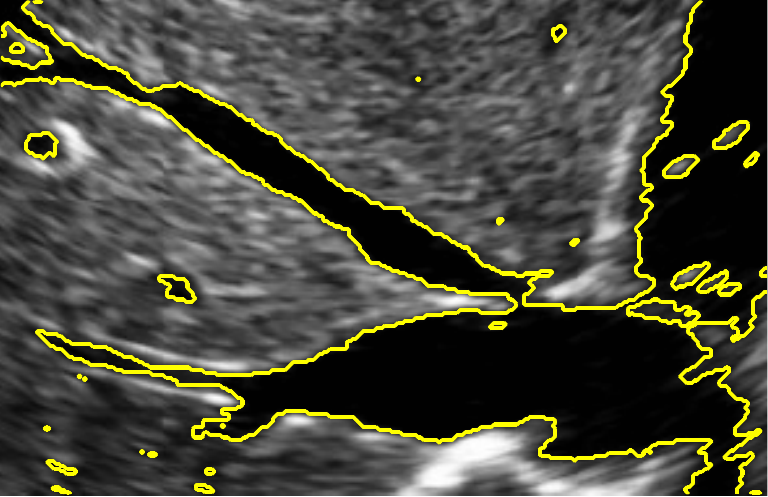}\\
\centering{(e) TV \eqref{mainmodelTV}}
\end{minipage}
\begin{minipage}[t]{4.5 cm}
\includegraphics[height=2.81 cm]{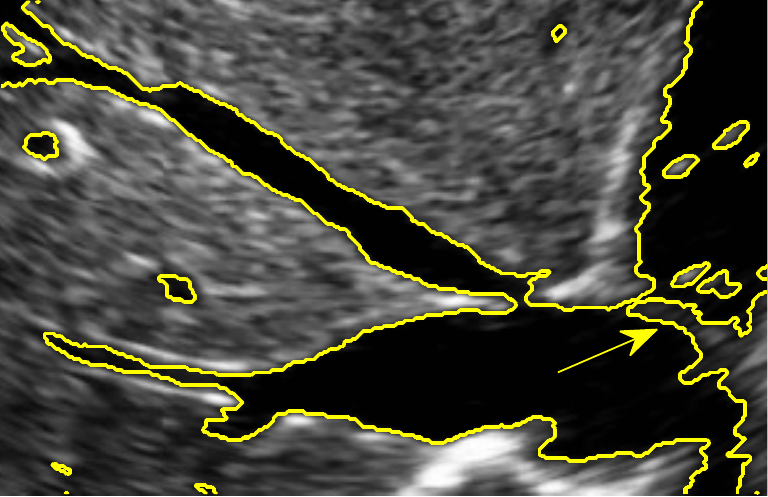}\\
\centering{(f) Tight-frame \eqref{mainmodel}}
\end{minipage}
\end{center}
\caption{\label{Liver}(a)--(d) results of \cite{CV}, \cite{YBT}, \cite{Li_MRI} and \cite{Kaihua2016level} respectively, (e) TV \eqref{mainmodelTV}, (f) Tight-frame \eqref{mainmodel}. }
\end{figure*}

\emph{Example \ref{Liver}}: Fig. \ref{Liver} is an image from an ultrasound data of a human liver. Fig. \ref{Liver}(a) from \cite{CV} fails to segment the liver as a whole, with many tiny holes left in the segmented region. Fig.~\ref{Liver}(b) from \cite{YBT} and (c) from \cite{Kaihua2016level} both get better overall segmentations than (a), but the lower left corners are not segmented well. Fig. \ref{Liver}(d) from \cite{Li_MRI} fails to produce a good segmentation. Fig. \ref{Liver}(e) and (f) from our methods both give satisfactory results, while the tight-frame regularisation gives more details of the lower right corner of the liver, please refer to the arrow in (f).

\begin{figure*}[htbp]
\begin{center}
\begin{minipage}[t]{4.5 cm}
\includegraphics[height=3 cm]{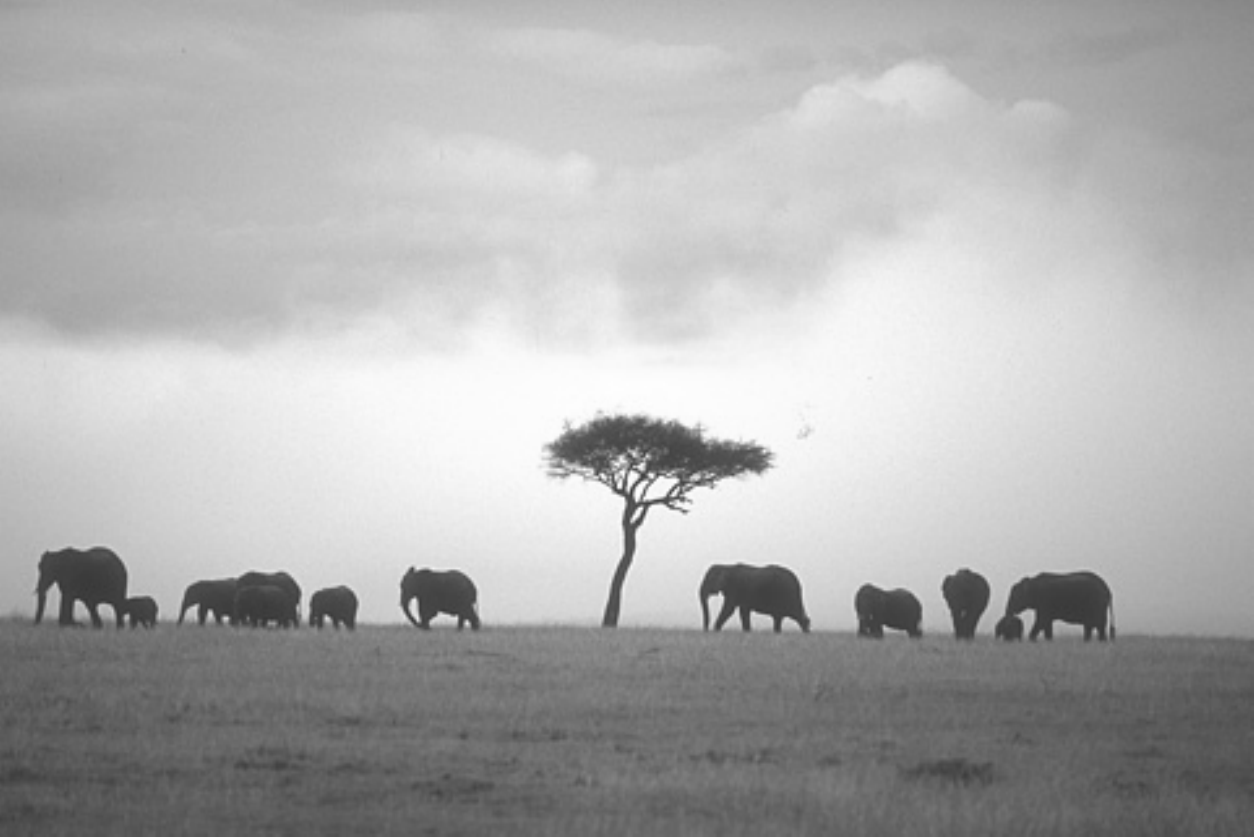}\\
\centering{(a) Original image}
\end{minipage}
\begin{minipage}[t]{4.5 cm}
\includegraphics[height=3 cm]{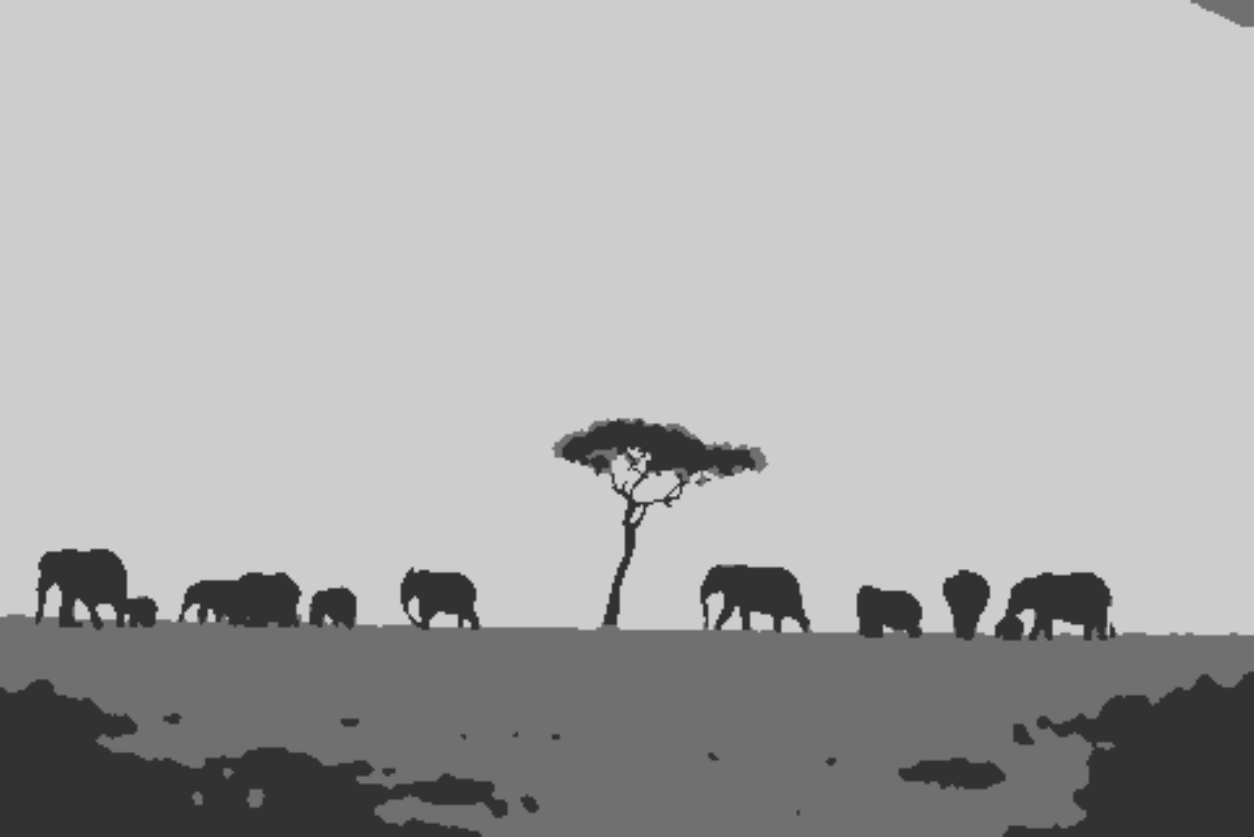}\\
\centering{(b) Yuan \cite{YBTBmul}}
\end{minipage}
\begin{minipage}[t]{4.5 cm}
\includegraphics[height=3 cm]{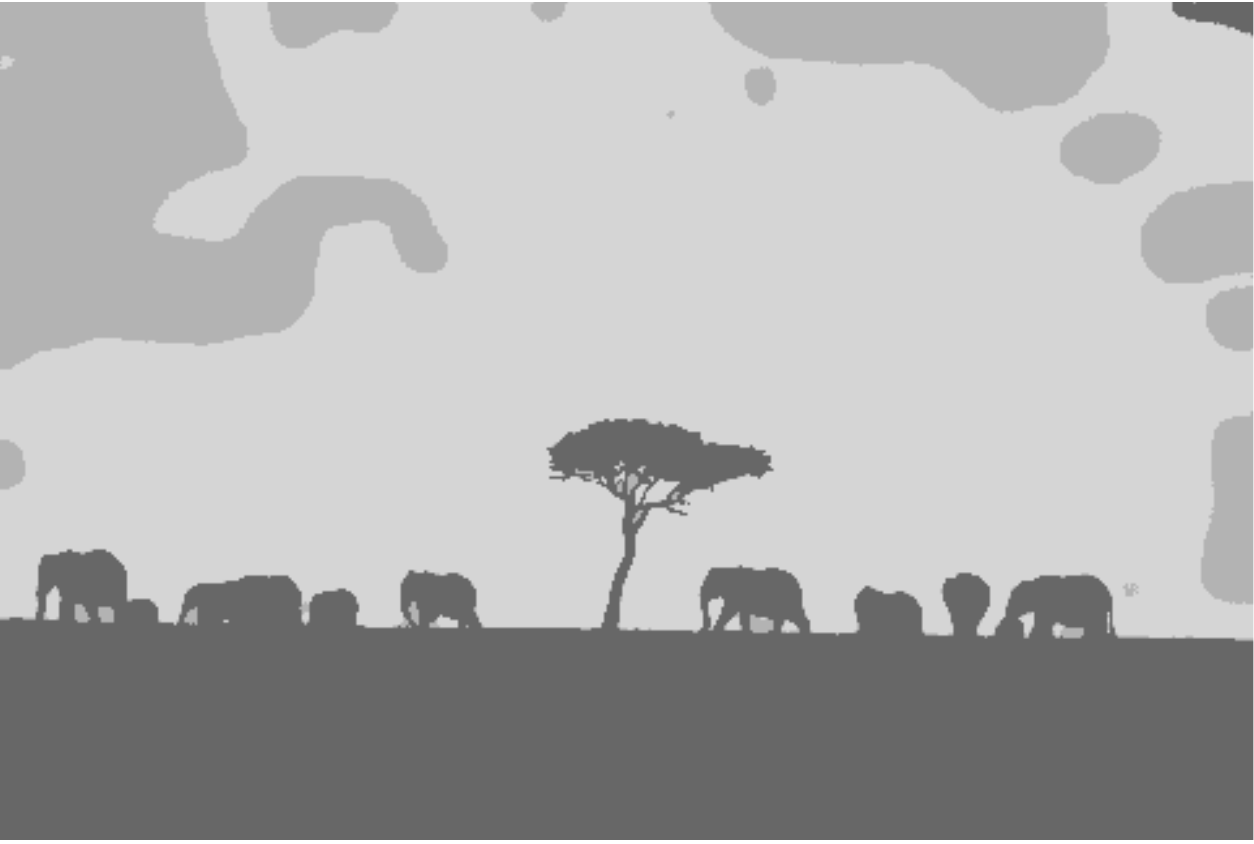}\\
\centering{(c) Li \cite{Li_MRI}}
\end{minipage}
\begin{minipage}[t]{4.5 cm}
\includegraphics[height=3 cm]{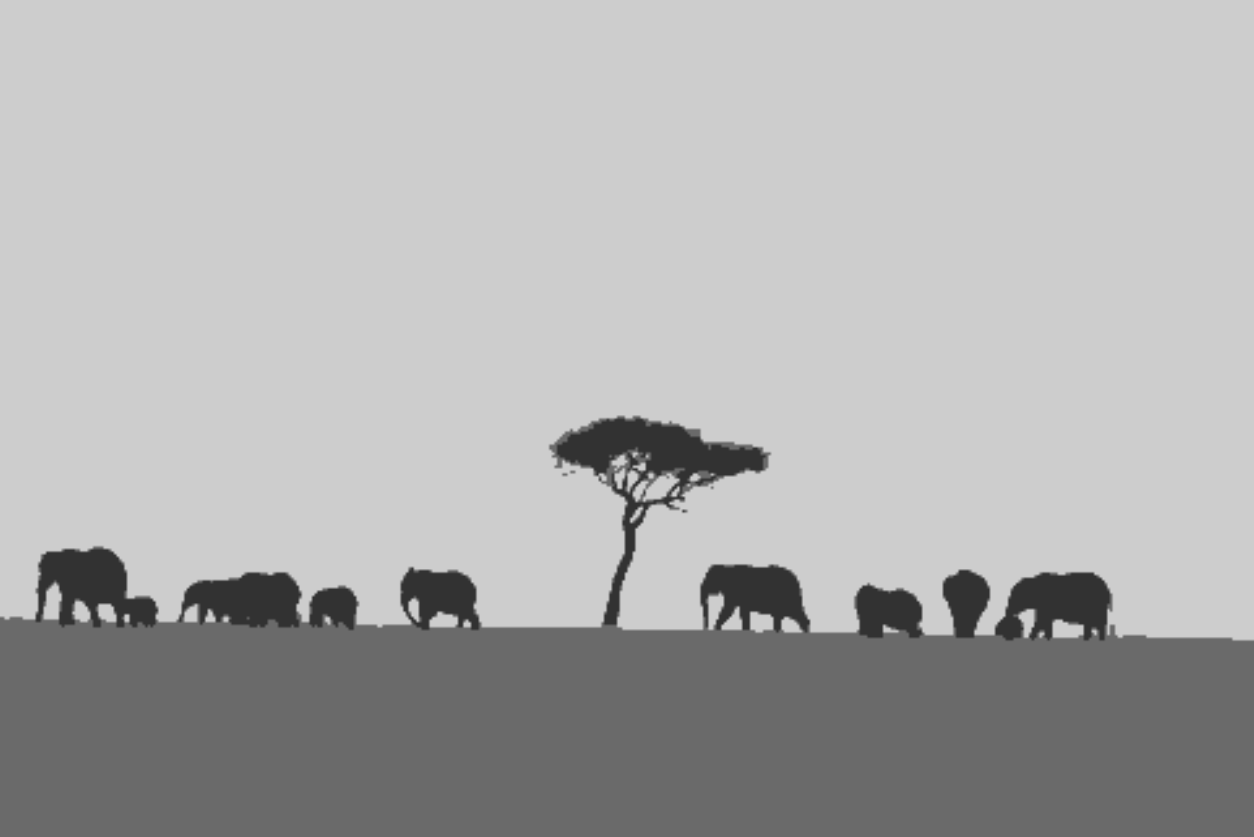}\\
\centering{(d) TV \eqref{mainmodelTV}}
\end{minipage}
\begin{minipage}[t]{4.5 cm}
\includegraphics[height=3 cm]{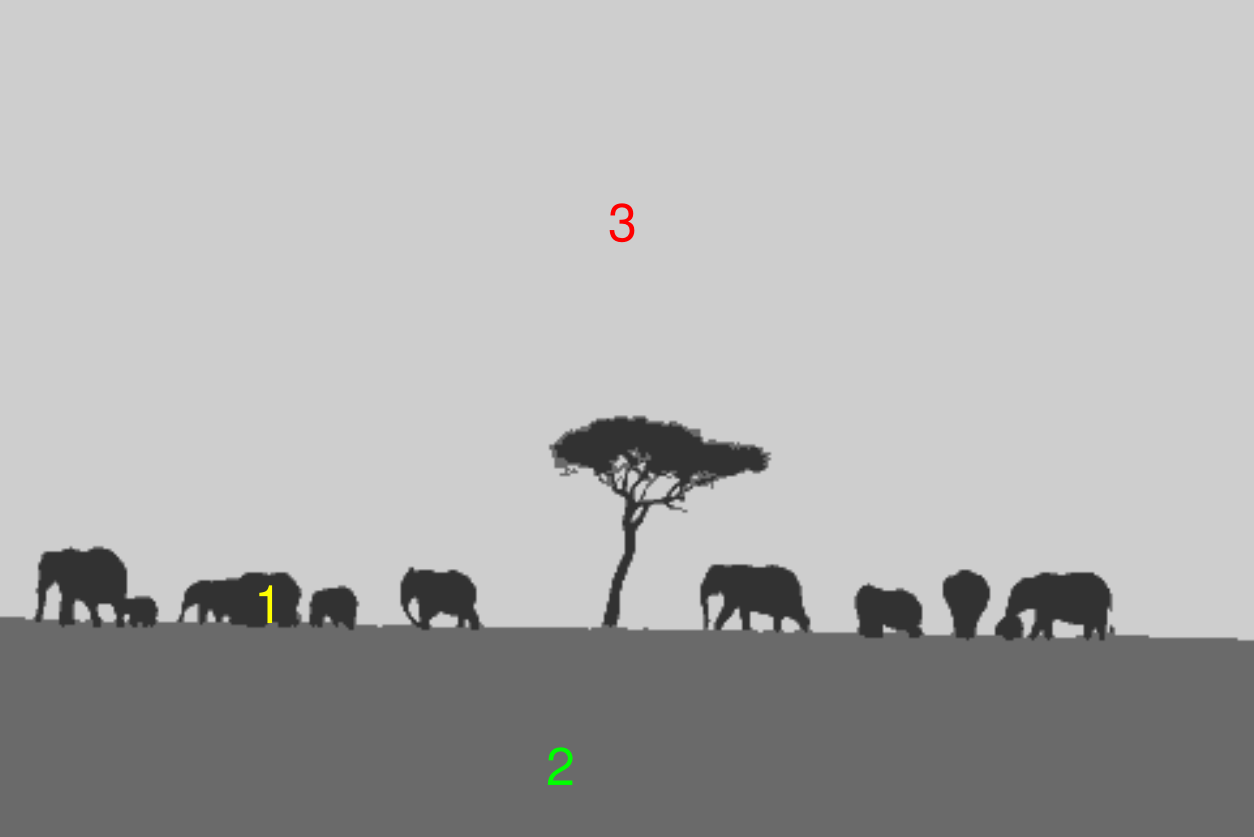}\\
\centering{(e) Tight-frame \eqref{mainmodel}}
\end{minipage}
\end{center}
\caption{\label{Animals}(a) Original image ``Animals'', (b) result of \cite{YBTBmul}, (c) result of \cite{Li_MRI}, (d)TV \eqref{mainmodelTV}, (e) Tight-frame \eqref{mainmodel}. }
\end{figure*}

\emph{Example \ref{Animals}}: This image is from the Berkeley Segmentation Dataset and Benchmark \cite{BerkleyDataset}. Fig. \ref{Animals}(a) is the original image. In this experiment, we try to segment Fig.~\ref{Animals}(a) into 3 phases: the ground, the tree and the elephants, and the sky. Fig.~\ref{Animals}(b) from \cite{YBTBmul} fails to segment the ground as a whole, and the upper right corner of the sky is segmented incorrectly. Fig.~\ref{Animals}(c) from \cite{Li_MRI} fails to distinguish the animals from the ground. Both Fig.~\ref{Animals}(d) and (e) from the TV \eqref{mainmodelTV} and tight frame \eqref{mainmodel} methods give good results, with the three phases correctly separated.

\begin{figure*}[htbp]
\begin{center}
\begin{minipage}[t]{2.6 cm}
\includegraphics[height=2.5 cm]{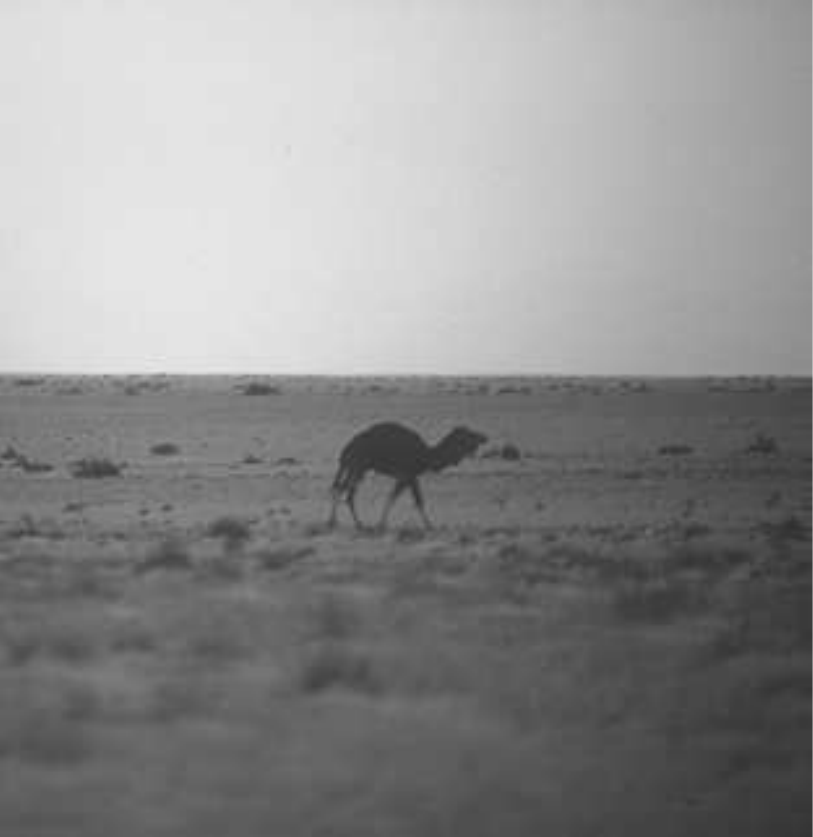}\\
\centering{(a) Original image}
\end{minipage}
\begin{minipage}[t]{2.6 cm}
\includegraphics[height=2.5 cm]{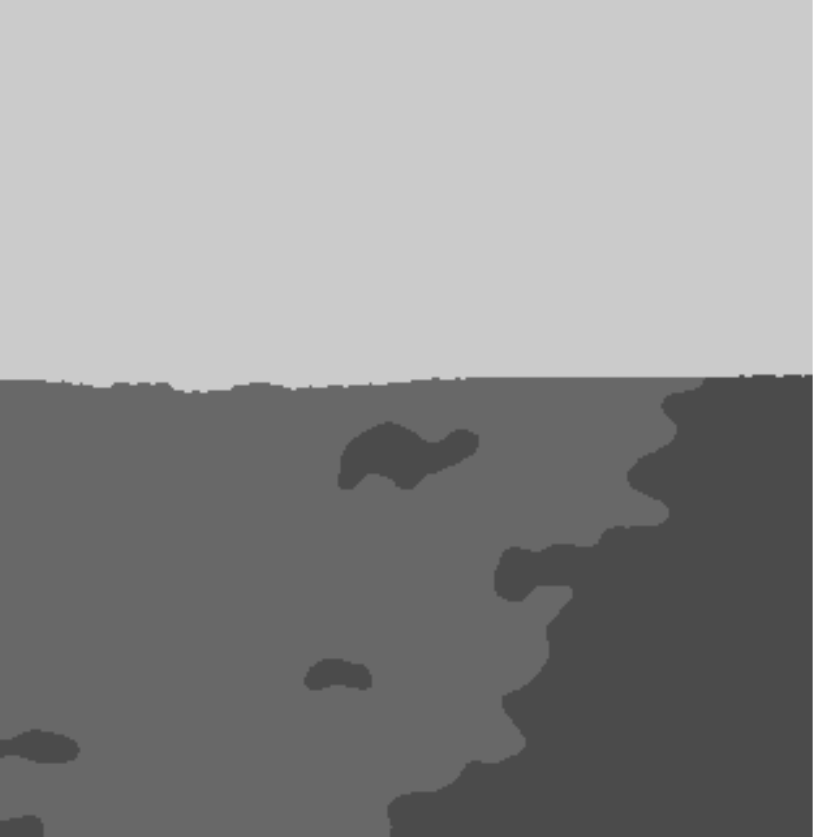}\\
\centering{(b) Yuan \cite{YBTBmul}}
\end{minipage}
\begin{minipage}[t]{2.6 cm}
\includegraphics[height=2.5 cm]{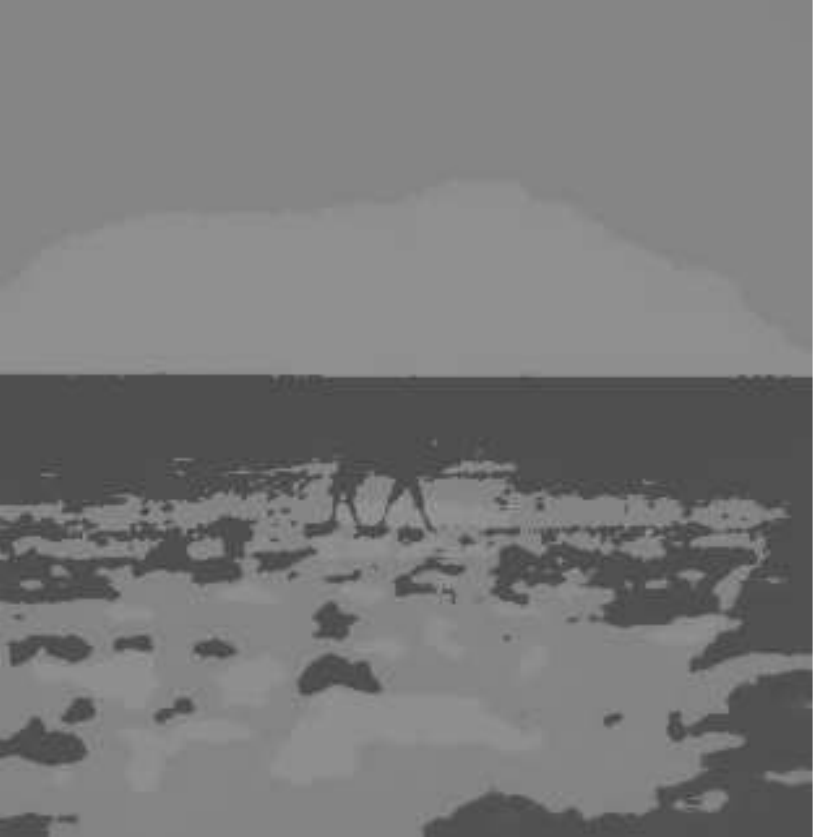}\\
\centering{(c) Li \cite{Li_MRI}}
\end{minipage}
\begin{minipage}[t]{2.6cm}
\includegraphics[height=2.5 cm]{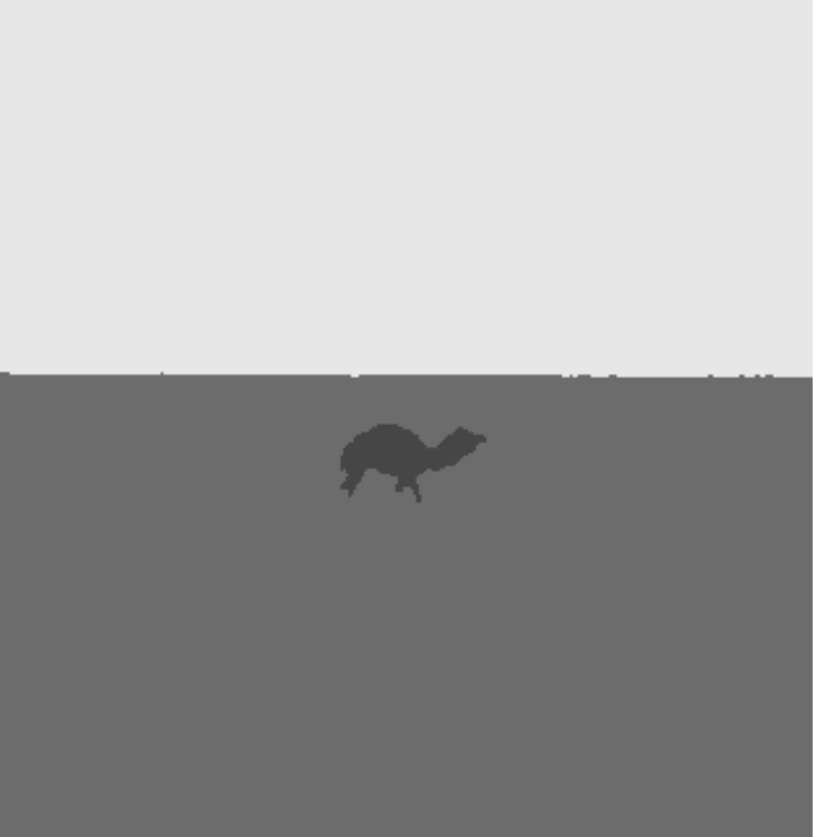}\\
\centering{(d) TV \eqref{mainmodelTV}}
\end{minipage}
\begin{minipage}[t]{2.6 cm}
\includegraphics[height=2.5 cm]{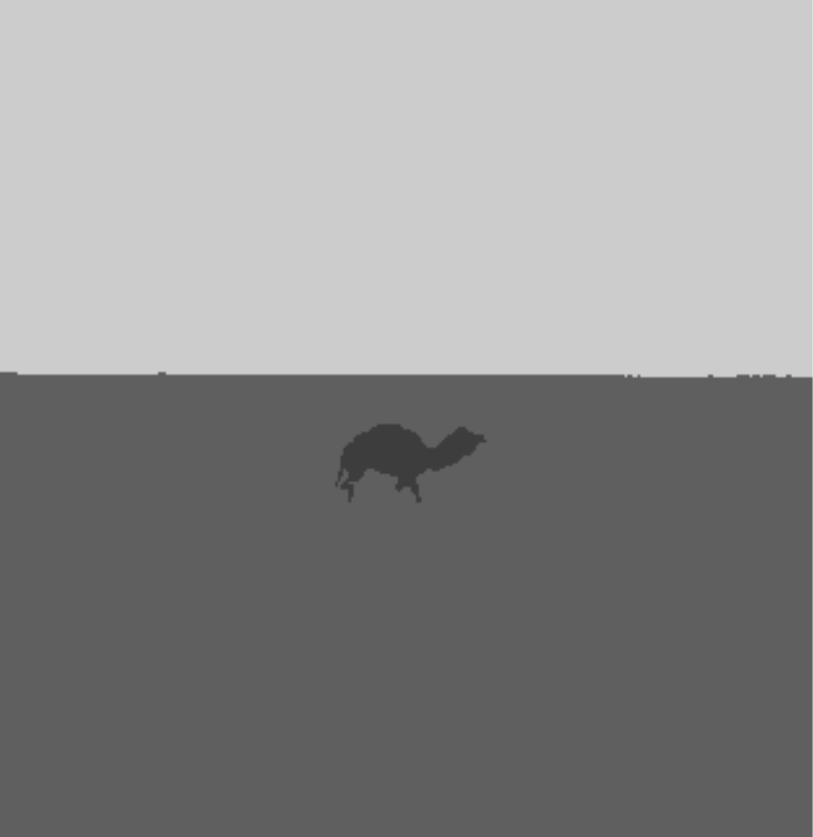}\\
\centering{(e) Tight-frame \eqref{mainmodel}}
\end{minipage}
\begin{minipage}[t]{4.0 cm}
\includegraphics[height=3.9 cm]{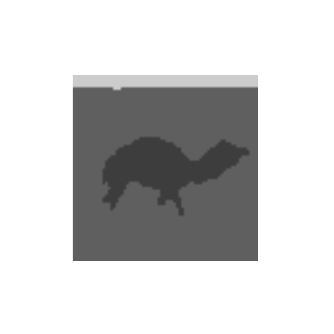}\\
\centering{(f) detail of (d)}
\end{minipage}
\begin{minipage}[t]{4.0 cm}
\includegraphics[height=3.9 cm]{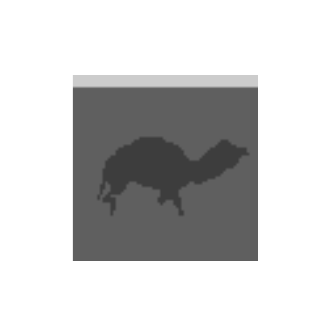}\\
\centering{(g) detail of (e)}
\end{minipage}
\end{center}
\caption{\label{Camel}(a) Original image ``Camel'', (b) result of \cite{YBTBmul}, (c) result of \cite{Li_MRI}, (d)TV \eqref{mainmodelTV}, (e) Tight-frame \eqref{mainmodel}, (f) detail of (d), (g) detail of (e). }
\end{figure*}

\emph{Example \ref{Camel}}: This image is also from the Berkeley Segmentation Dataset and Benchmark \cite{BerkleyDataset}. Fig.~\ref{Camel}(a) is the original (cropped) image. We aim to segment this image into three phases. It is clear that this image has severe intensity inhomogeneity: the right part of the image is darker, with the bottom right corner having similar intensity values as that of the camel in the middle of the image. Fig.~\ref{Camel}(b) from \cite{YBTBmul} manages to separate the camel, but a large part of the ground is wrongly segmented. Fig.~\ref{Camel}(c) fails to give a reasonable segmentation. Fig.~\ref{Camel}(d) and (e) from our TV \eqref{mainmodelTV} and tight frame \eqref{mainmodel} models both give good segmentations, with clear separations of the camel, the ground, and the sky. From the details in Fig.~\ref{Camel}(f) and (g), we see that tight frame regularisation can even preserve the tail of the camel while the TV regularisation fails to do so.

\begin{figure*}[htbp]
\begin{center}
\begin{minipage}[t]{4.3 cm}
\includegraphics[height=3 cm]{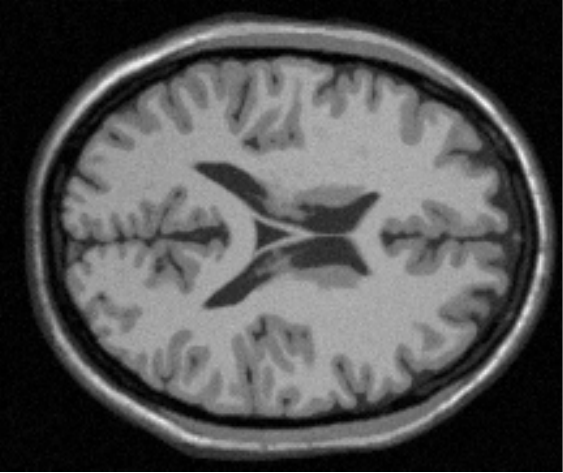}\\
\centering{(a) Original image}
\end{minipage}
\begin{minipage}[t]{4.3 cm}
\includegraphics[height=3 cm]{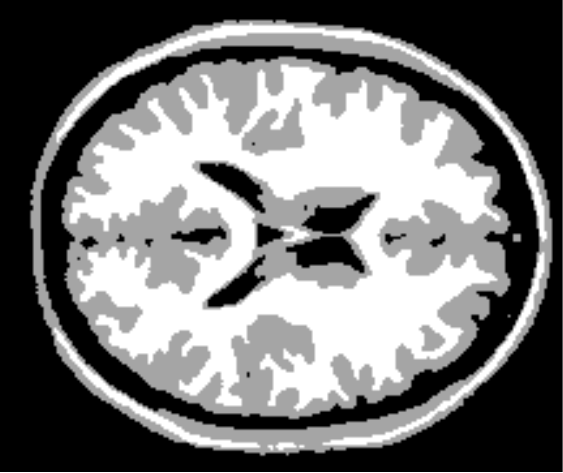}\\
\centering{(b) Yuan \cite{YBTBmul} }
\end{minipage}
\begin{minipage}[t]{4.3 cm}
\includegraphics[height=3 cm]{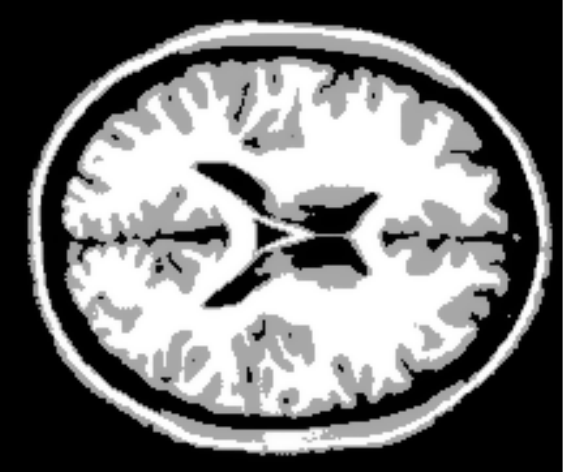}\\
\centering{(c) Li \cite{Li_MRI}}
\end{minipage}
\begin{minipage}[t]{4.3 cm}
\includegraphics[height=3 cm]{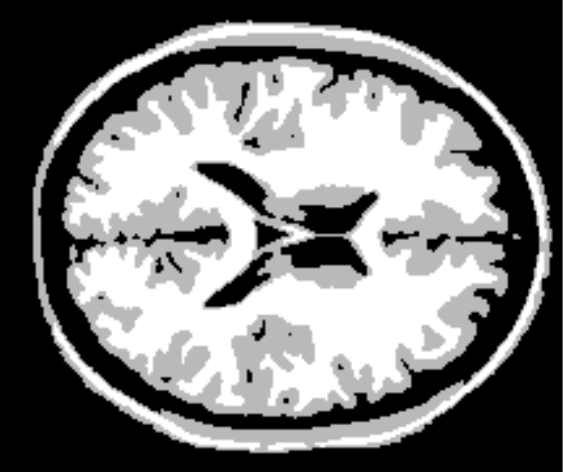}\\
\centering{(d) TV \eqref{mainmodelTV}}
\end{minipage}
\begin{minipage}[t]{4.3 cm}
\includegraphics[height=3 cm]{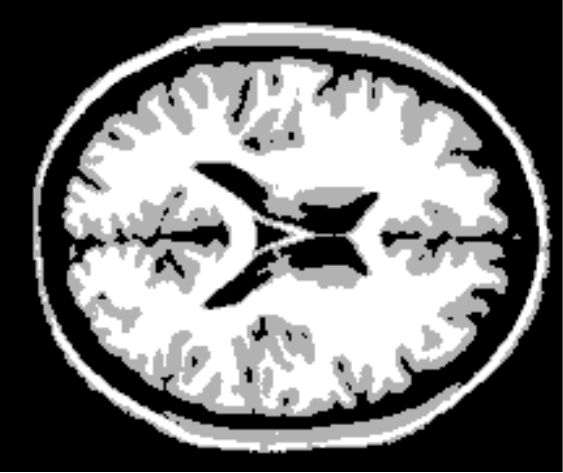}\\
\centering{(e) Tight-frame \eqref{mainmodel}}
\end{minipage} \\
\begin{minipage}[t]{2.5 cm}
\includegraphics[height=2.5 cm]{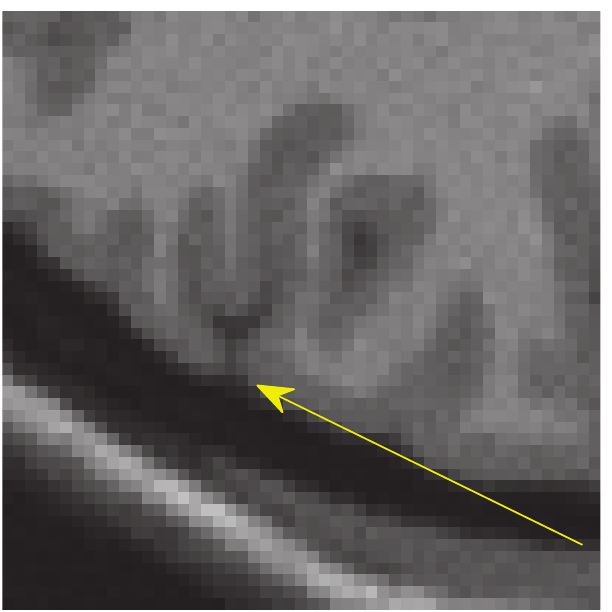}\\
\centering{(f) Detail of (a)}
\end{minipage}
\begin{minipage}[t]{2.5 cm}
\includegraphics[height=2.5 cm]{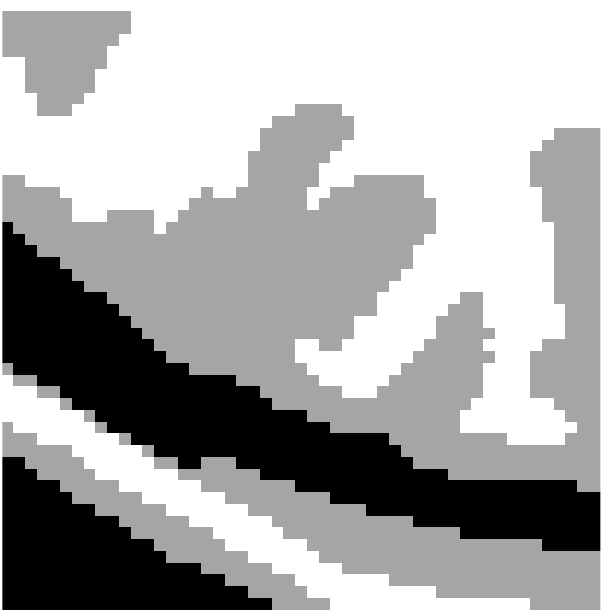}\\
\centering{(g) Detail of (b)}
\end{minipage}
\begin{minipage}[t]{2.5 cm}
\includegraphics[height=2.5 cm]{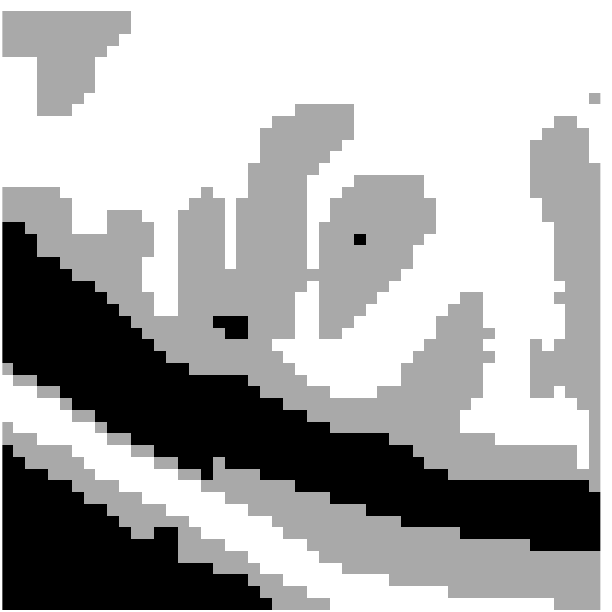}\\
\centering{(h) Detail of (c)}
\end{minipage}
\begin{minipage}[t]{2.5 cm}
\includegraphics[height=2.5 cm]{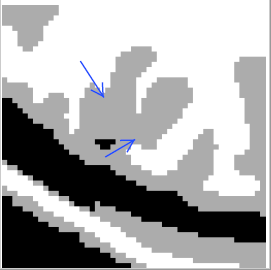}\\
\centering{(i) Detail of (d)}
\end{minipage}
\begin{minipage}[t]{2.5 cm}
\includegraphics[height=2.5 cm]{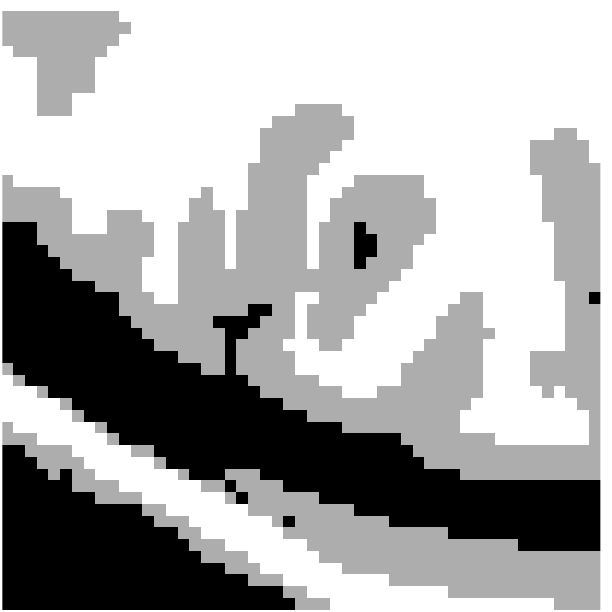}\\
\centering{(j) Detail of (e)}
\end{minipage}
\end{center}
\caption{\label{Brain}(a) Original image ``Brain'', (b) result of \cite{YBTBmul} (c) result of \cite{Li_MRI}, (d)TV \eqref{mainmodelTV}, (e) Tight-frame \eqref{mainmodel}, (f)--(j) details of (a)--(e). }
\end{figure*}

\emph{Example \ref{Brain}}: Lastly we show an example of multi-phase segmentation on a simulated Brain MRI image. This image is obtained from \url{http://www.bic.mni.mcgill.ca/brainweb/}, with T1 modality, 1mm slice thickness, 3\% noise and 20\% intensity non-uniformity. In this experiment, our goal is to separate the background, the gray matter of the brain, and the white matter of the brain. From the detailed images Fig.~\ref{Brain}(f)--(j), we see that our method with tight frame regularisation \eqref{mainmodel} gives the best result. Please note the crack indicated by an arrow in Fig.~\ref{Brain}(f). Only Fig.~\ref{Brain}(j) from tight frame regularisation manages to reserve the crack. We notice that compared with the TV regularisation, the tight frame regularisation can get more details in the white matter of the brain (see the two arrows in Fig.~\ref{Brain}(i) which indicate missing details).

\section{Conclusion and possible future improvements}

In this paper, we have proposed a method to segment images with intensity inhomogeneity. We use both TV and tight-frame regularisation in our method to explore their difference and connection. There are two-stages in the segmentation: in the first stage, we solve a convex minimization problem to decouple the original image into reflection and illumination, and in the second stage we segment the image by thresholding the reflection part of the image. Comparing with \cite{TVRetinex}, which is for image enhancement, our model has an extra smoothing term on the reflection part to efface  tiny structures in images. Moreover, we propose a unified primal-dual method to solve our model with both TV and tight-frame regularisation. This is better than the inexact approach appearing in \cite{TVRetinex} since the convergence of our numerical scheme is guaranteed. Furthermore, numerical experiments show that our approach can produce good segmentations for various images with intensity inhomogeneity, and the introduction of framelet regularisation improves fine details of the segmentations.

Our approach have several advantages. First, our models \eqref{mainmodelTV} and \eqref{mainmodel} in the first stage are convex, which guarantees the uniqueness of the solutions and the stability of our algorithm.
The proposed numerical scheme has only one loop and each step is exact. Secondly, our formulation of intensity inhomogeneity can segment both natural and medical images well. Thirdly, in the second stage, the choice of a threshold and the number of phases are independent of the first stage. Therefore, our method is capable of segmentations with any number of phases, and users can try different number of phases or thresholds without recalculating the first stage.

Our method can be further improved in several ways. One is to construct more efficient algorithms to solve model \eqref{mainmodel}. Here in order to get satisfactory results, we set a large iteration number in the implementation of the primal-dual algorithm to solve model \eqref{mainmodel}. Another possible improvement is to consider automatic clustering algorithms, like the K-means method \cite{HW} or the DBSCAN method \cite{Ester_1996}, to determine the thresholds in the second stage.

\end{document}